 \documentclass[letterpaper,11pt]{article}
\usepackage{etoolbox}
\newtoggle{anonymous}
\togglefalse{anonymous}

\usepackage[usenames,dvipsnames,svgnames,table]{xcolor}
\usepackage[colorlinks=true,linkcolor=blue,citecolor=ForestGreen]{hyperref}
\usepackage{amsmath,amssymb,enumerate}
\usepackage{scrextend}
\usepackage[utf8]{inputenc}
\usepackage[T1]{fontenc}

\usepackage{paralist}
\usepackage{bbm}
\usepackage{hyperref}
\usepackage{float}
\usepackage[color=green!40,textsize=small]{todonotes}
\usepackage{xcolor}
\usepackage{forest}
\usepackage{fixmath}

\usepackage{amsthm}
\usepackage[capitalize]{cleveref}

\makeatletter
\newtheorem*{rep@theorem}{\rep@title}
\newcommand{\newreptheorem}[2]{%
  \newenvironment{rep#1}[1]{%
    \def\rep@title{#2 \ref{##1}}%
    \begin{rep@theorem}}%
  {\end{rep@theorem}}
}
\makeatother

\newreptheorem{theorem}{Theorem}

 \usepackage[margin=1in]{geometry}

\usepackage{setspace}
\allowdisplaybreaks

\Crefname{figure}{Figure}{Figures}
\Crefname{claim}{Claim}{Claims}

\crefformat{equation}{\textup{#2(#1)#3}}
\crefrangeformat{equation}{\textup{#3(#1)#4--#5(#2)#6}}
\crefmultiformat{equation}{\textup{#2(#1)#3}}{ and \textup{#2(#1)#3}}
{, \textup{#2(#1)#3}}{, and \textup{#2(#1)#3}}
\crefrangemultiformat{equation}{\textup{#3(#1)#4--#5(#2)#6}}%
{ and \textup{#3(#1)#4--#5(#2)#6}}{, \textup{#3(#1)#4--#5(#2)#6}}{, and \textup{#3(#1)#4--#5(#2)#6}}

\Crefformat{equation}{Equation~\textup{#2(#1)#3}}
\Crefrangeformat{equation}{Equations~\textup{#3(#1)#4--#5(#2)#6}}
\Crefmultiformat{equation}{Equations~\textup{#2(#1)#3}}{ and \textup{#2(#1)#3}}
{, \textup{#2(#1)#3}}{, and \textup{#2(#1)#3}}
\Crefrangemultiformat{equation}{Equations~\textup{#3(#1)#4--#5(#2)#6}}%
{ and \textup{#3(#1)#4--#5(#2)#6}}{, \textup{#3(#1)#4--#5(#2)#6}}{, and \textup{#3(#1)#4--#5(#2)#6}}

\usepackage{xspace}
\newcommand*\ie{i.\kern.1em e.\@\xspace} 
\newcommand*\eg{e.\kern.1em g.\@\xspace}

\usepackage{tikz}
\tikzset{normalnode/.style={circle, draw, fill=black, inner sep=0, minimum width=1.5mm}}

\usetikzlibrary{fit}
\usetikzlibrary{shapes,arrows,shadows,positioning}
\usetikzlibrary{backgrounds}
\usetikzlibrary{calc}

\definecolor{agreen}{rgb}{.24,.86,.53}

\newtheorem{theorem}{Theorem}[section]

\newtheorem{proposition}[theorem]{Proposition}

\newtheorem{observation}[theorem]{Observation}

\newtheorem{lemma}[theorem]{Lemma}

\theoremstyle{definition}

\theoremstyle{remark}

\newtheorem{claim}[theorem]{Claim}

\newenvironment{poc}{\begin{proof}[Proof of {C}laim.]}{\end{proof}}

\crefname{observation}{Observation}{Observations}

\renewcommand{\leq}{\leqslant}
\renewcommand{\geq}{\geqslant}

\newcommand{\comment}[1]{}

\newcommand{\sd}{\text{sd}}
\newcommand{\sdd}{\text{sdd}}

\title{Symmetric-Difference (Degeneracy) and Signed Tree Models}

\date{}

 \iftoggle{anonymous}{% ANONYMOUS
 	\author{ }
 }{%
	\author{{\'E}douard Bonnet\thanks{Univ. Lyon, ENS de Lyon, UCBL, CNRS, LIP, France,  \texttt{edouard.bonnet@ens-lyon.fr}}    \and Julien Duron\thanks{Univ. Lyon, ENS de Lyon, UCBL, CNRS, LIP, France, \texttt{julien.duron@ens-lyon.fr}} \and    John Sylvester\thanks{Department of Computer Science, University of Liverpool, UK, \texttt{john.sylvester@liverpool.ac.uk}}
	\and
	Viktor Zamaraev\thanks{Department of Computer Science, University of Liverpool, UK, \texttt{viktor.zamaraev@liverpool.ac.uk}} }
}

\begin{document}
\maketitle

\begin{abstract}
  We introduce a dense counterpart of graph degeneracy, which extends the recently-proposed invariant symmetric difference.
  We say that a graph has \emph{sd-degeneracy} (for symmetric-difference degeneracy) at most $d$ if it admits an elimination order of its vertices where a vertex~$u$ can be removed whenever it has a $d$-twin, i.e., another vertex $v$ such that at most $d$ vertices outside $\{u,v\}$ are neighbors of exactly one of $u, v$.   
  The family of graph classes of bounded \mbox{sd-degeneracy} is a superset of that of graph classes of bounded degeneracy or of bounded flip-width, and more generally, of bounded symmetric difference.
  Unlike most graph parameters, sd-degeneracy is not hereditary: it may be strictly smaller on a~graph than on some of its induced subgraphs.
  In particular, every $n$-vertex graph is an induced subgraph of some $O(n^2)$-vertex graph of \mbox{sd-degeneracy}~1.
  In spite of this and the breadth of classes of bounded sd-degeneracy, we devise $\tilde{O}(\sqrt n)$-bit adjacency labeling schemes for them, which are optimal up to the hidden polylogarithmic factor. 
  This is attained on some even more general classes, consisting of graphs~$G$ whose vertices bijectively map to the leaves of a~tree~$T$, where transversal edges and anti-edges added to~$T$ define the edge set of~$G$.
  We call such graph representations \emph{signed tree models} as they extend the so-called tree models (or twin-decompositions) developed in the context of twin-width, by adding transversal anti-edges.
  
  While computing the degeneracy of a~graph takes linear time, we show that determining its symmetric difference is para-co-NP-complete.
  This may seem surprising as symmetric difference can serve as a short-sighted first approximation of twin-width, whose computation is para-NP-complete.
  Indeed, we show that deciding if the symmetric difference of an input graph is at~most~8 is co-NP-complete.
  We also show that deciding if the sd-degeneracy is at~most~1 is NP-complete, contrasting with the symmetric difference.
\end{abstract}

\section{Introduction}\label{sec:intro}

There are two theories of sparse graphs: the so-called Sparsity Theory pioneered by Nešetřil and Ossona de Mendez~\cite{sparsity}, and the theory behind the equivalent notions of bounded degeneracy, maximum average degree, subgraph density, and arboricity.
One of the many merits of the former theory is to capture efficient first-order model checking within subgraph-closed classes, with the so-called \emph{nowhere dense} classes~\cite{Grohe17}.
\emph{Monadic stability} constitutes a~dense analogue of nowhere denseness with similar algorithmic properties~\cite{Dreier23}.
The second theory has, as we will see, simple but useful connections with the chromatic number and adjacency labeling schemes. 
One of our two main motivations is to introduce and explore dense analogues of it.

The degeneracy of a~graph~$G$ is the minimum integer~$d$ such that every induced subgraph of~$G$ has a~vertex of degree at~most~$d$. %, or equivalently, every $k$-vertex subgraph of $G$ has at most $dk/2$ edges.
As removing vertices may only decrease the degree of the remaining ones, checking that the degeneracy is at~most~$d$ can be done greedily.
This prompts the following equivalent definition, equal to the \emph{coloring number}\footnote{not to be confused with the \emph{chromatic} number} minus one~\cite{Erdos66,Die2017}.
A~graph has degeneracy at~most~$d$ if there is a~total order, called \emph{degeneracy ordering}, on its vertices such that every vertex $v$ has at most $d$ neighbors following $v$ in the order.
The degeneracy is then the least integer~$d$ such that an ordering witnessing degeneracy at~most~$d$ exists.
Given such an ordering, a~graph can be properly $(d+1)$-colored by a~greedy strategy: use the smallest available color looping through vertices in the reverse order.
Another advantage of the definition via degeneracy ordering is that it yields a~polynomial-time algorithm to compute the degeneracy.
While the graph is nonempty, find a~vertex of minimum degree, append it to the order, and remove it from the graph.
Degeneracy is frequently used to bound the chromatic number from above. 
For instance, until recently~\cite{Norin23} the Kostochka--Thomason degeneracy bound of graphs without $K_t$ minor~\cite{Kostochka84,Thomason84} was the best way we knew of coloring these graphs.   

%put a vertex of minimum degree first in the order and remove it from the graph, put a vertex of the minimum degree in the remaining graph second in the order and remove it from the graph, and so on until all vertices are in the order. An important consequence of the latter property is that degeneracy can be computed in linear time.
%We note that degeneracy is a hereditary graph parameter in the sense that it does not increase when taking induced subgraphs.
%Furthermore, a simple but important observation is that a degeneracy ordering of width $d$ is inherited by induced subgraphs, \ie,
%if $\pi$ is a degeneracy ordering of width $d$ for $G$, then, for any induced subgraph $H$ of $G$, the restrictions of $\pi$ to the vertices of $H$ is a degeneracy ordering of width $d$ for $H$.

Another application of bounding the degeneracy is to obtain implicit representations.
Indeed, graphs of bounded degeneracy admit $f(n)$-bit \emph{adjacency labeling schemes} with $f(n)=O(\log n)$.\footnote{Throughout the paper, $\log$ denotes the logarithm function in base 2.}
In other words, given a class of graphs of degeneracy at~most~$d$, there exists an algorithm, called \emph{decoder}, such that the vertices of any $n$-vertex graph $G$ from the class can be assigned \emph{labels} (which are binary strings) of length $f(n)$ in such a way that the decoder can infer the adjacency of any two vertices $u, v$ in $G$ from their mere labels. 
An~$O(\log n)$-bit labeling scheme is easy to design for any class $\mathcal C$ of bounded degeneracy.
From an ordering of $G \in \mathcal C$ witnessing degeneracy~$d$, the label of each vertex stores its own index in the ordering and the indices of its at most $d$~neighbors that follow it in the order.
Then the decoder just checks whether the index of one of $u, v$ is among the indices of the neighbors of the other vertex.
Note that each label has size at~most $(d+1)\lceil \log n \rceil$.
For example, this was recently used to show that every subgraph-closed class with single-exponential speed of growth admits such a~labeling scheme~\cite{BDSZZ23mono}.

Adjacency labeling schemes of size $O(\log n)$ are at the heart of the recently-refuted Implicit Graph Conjecture (IGC)~\cite{KNR88,Mul88}.
The IGC speculated that the information-theoretic necessary condition for a~hereditary graph class to have an~$O(\log n)$-bit labeling scheme is also sufficient.
This necessary condition comes from the observation that a string of length $O(n \log n)$ obtained by concatenating all vertex labels is an encoding of the graph.
Therefore a~class of graphs that admits an adjacency labeling scheme of size $O(\log n)$ contains at most $2^{O(n \log n)}$ (un)labeled $n$-vertex graphs.
Graph classes with such a~bound on the number of (un)labeled $n$-vertex graphs are called~\emph{factorial}. 
In this terminology, the IGC can be stated as follows: any hereditary factorial graph class admits an $O(\log n)$-bit adjacency labeling scheme.

The IGC has been refuted by a~wide margin;
in a breakthrough work, Hatami and Hatami \cite{HH22} showed that there are factorial hereditary graph classes for which any adjacency labeling scheme requires labels of length~$\Omega(\sqrt{n})$.
However, the refutation is based on a~counting argument and does not pinpoint an explicit counterexample.
There are a~number of explicit factorial graph classes that could refute the IGC, but the conjecture is still open for these classes.
Let us call EIGC (for Explicit Implicit Graph Conjecture) this very challenge.
For instance, whether the IGC holds within intersection graphs of segments, unit disks, or disks in the plane, and more generally semi-algebraic graph classes, is unsettled.
Despite the workable definitions of these classes, the geometric representations alone cannot lead to $O(\log n)$-bit labeling schemes~\cite{MM13}.
If such labeling schemes exist, they are likely to utilize some non-trivial structural properties of these graphs.
%The problem with this approach is that very little is known about the structure these geometric graph classes. 

The graph parameter \emph{symmetric difference} was introduced to design a~candidate to explicitly refute the IGC \cite{ACLZ15}.
A~graph $G$ has symmetric difference at~most~$d$ if in every induced subgraph of $G$ there is a pair of vertices $u,v$ such that there are at most $d$ vertices different from $u$ and $v$ that are adjacent to exactly one of $u,v$.
In other words, $u$ and $v$ are \emph{$d$-twins}, i.e., they become twins after removing at~most~$d$ vertices from the graph.
One can construe symmetric difference as a~dense analogue of the first definition of degeneracy given above.
Symmetric difference is a~hereditary graph parameter: it can only decrease when taking induced subgraphs.
Like classes of bounded degeneracy, classes of bounded symmetric difference are factorial~\cite{ACLZ15}.
Symmetric difference generalizes degeneracy in the sense that any class of graphs of bounded degeneracy has bounded symmetric difference.
Indeed, if a~graph has degeneracy at~most~$d$, then it has symmetric difference as~most~$2d$: for any graph with an ordering witnessing degeneracy~$d$, the first two vertices in the order are $2d$-twins.
Notice that, on the other hand, complete graphs have unbounded degeneracy, but their symmetric difference is~0.
The existence of an $O(\log n)$-bit adjacency labeling scheme for graphs of bounded symmetric difference remains open.

%Motivated by this later challenge, in the present paper, we study symmetric difference in more detail.
%In particular, we reveal more differences between the two parameters and design $O(\sqrt{n} \cdot \polylog(n))$ adjacency labeling schemes for classes of graphs of bounded symmetric difference, which, for what we know from \cite{HH22}, might be best possible up to $\polylog(n)$ factor.
\medskip
\textbf{Our contribution.}
We introduce another dense analogue of degeneracy based on the second given definition.  
The \emph{sd-degeneracy} (for \emph{symmetric-difference degeneracy}) of a~graph~$G$ is the least integer~$d$ for which there is an ordering of the vertices of $G$ such that every vertex $v$ but the last one admits a~$d$-twin in the subgraph of~$G$ induced by $v$ and all the vertices following it in the order. 
It follows from the definitions that graphs with sd-degeneracy at~most~$d$ form a~superset of graphs with symmetric difference at~most~$d$.
Contrary to what happens in the sparse setting with degeneracy, this superset is strict. In fact, there are classes with sd-degeneracy~$1$ and unbounded symmetric difference.
\begin{proposition}\label{prop:anygraphinSD1}
	For any $n$-vertex graph $G$, there exists a graph of sd-degeneracy~$1$ with less than~$n^2$ vertices containing $G$ as an induced subgraph.  
\end{proposition}
\begin{proof}
  Take any $n$-vertex graph $G$, and fix an arbitrary ordering $v_1, v_2, \ldots, v_n$ of its vertices.
  There is a~graph~$G'$ with at~most~$(n-1)(n-3)$ vertices that both contains $G$ as an induced subgraph, and has \mbox{sd-degeneracy} at~most~$1$.
The graph~$G'$ can be built by adding to $G$, for every $i \in [n-1]$, up~to~$n-3$ vertices $v_{i,1}, \ldots, v_{i,h_i}$ gradually ``interpolating'' between the neighborhood of $v_i$ and that of $v_{i+1}$ (in~$G$).
For instance, $v_{i,1}, \ldots, v_{i,h'_i}$ remove one-by-one the neighbors of $v_i$ that are not neighbors of~$v_{i+1}$, and $v_{i,h'_i+1}, \ldots, v_{i,h_i}$ add one-by-one the neighbors of $v_{i+1}$ that are not neighbors of~$v_i$.
(We put no edge between a~pair of added vertices $v_{i,p}, v_{i',p'}$.)
Then an ordering witnessing sd-degeneracy at~most~1 is $v_1, v_{1,1}, v_{1,2}, \ldots, v_{1,h_1}, v_2, v_{2,1}, v_{2,2}, \ldots, v_{2,h_2}, \ldots, v_{n-1}, v_{n-1,1}, v_{n-1,2}, \ldots, v_{n-1,h_{n-1}}, v_n$, for which the 1-twin of a~vertex is its successor.
\end{proof}

By an aforementioned counting argument, the class of all graphs requires labeling schemes of size $\Theta(n)$.
Therefore, by \cref{prop:anygraphinSD1}, the (non-hereditary) class of graphs with sd-degeneracy at~most~1 requires adjacency labels of size $\Omega(\sqrt n)$.
Surprisingly, we match this lower bound with a~labeling scheme, tight up to a polylogarithmic factor, for any class of bounded sd-degeneracy.

\begin{theorem}\label{thm:als-sdd}
  The class of all graphs with sd-degeneracy~at~most~$d$ admits an~$O(\sqrt{dn}\log^3 n)$-bit adjacency labeling scheme.
\end{theorem}

The tool behind the proof of~\cref{thm:als-sdd} is the second motivation of the paper.
We wish to unify and extend twin-decompositions of low width (also called tree models)~\cite{twin-width3,BNO21} developed in the context of twin-width, and spanning paths (or Welzl orders) of low crossing number (or low alternation number)~\cite{Welzl88}, which are useful orders in answering geometric range queries; 
these orders also appear in the context of implicit graph representations~\cite{Alon23,AAALZ23}, and were utilized as part of efficient first-order model checking algorithms~\cite{Dreier23}.
%also see~\cite{Dreier23} which utilizes these orders as part of efficient first-order model checking algorithms.
We thus introduce \emph{signed tree models}.
A~signed tree model of a~graph~$G$ is a~tree whose leaves are in one-to-one correspondence with the vertices of~$G$, together with extra transversal edges and anti-edges, which fully determine (see the exact rules in~\cref{sec:stm}) the edges of~$G$.
The novelty compared to the existing tree models is the presence of transversal anti-edges.
We show that graphs with signed tree models of degeneracy at~most~$d$ admit a~labeling scheme as in~\cref{thm:als-sdd}.
The latter theorem is then obtained by building such a~signed tree model for any graph of sd-degeneracy~$d$.

When given the vertex ordering witnessing sd-degeneracy~$d$, the labeling scheme can be effectively computed.
However, we show that computing the sd-degeneracy of a~graph (hence, in particular a~witnessing order) is NP-complete, even when the sd-degeneracy is guaranteed to be below a~fixed constant.
In the language of parameterized complexity, sd-degeneracy is para-NP-complete.

\begin{theorem}\label{thm:sdd-npc}
  Deciding if a~graph has sd-degeneracy at~most~1 is NP-complete.
\end{theorem}

This is in stark contrast with the existing linear-time cograph recognition~\cite{Corneil85}, which is equivalent to determining if the sd-degeneracy is (at~most)~0. 
We show that, surprisingly, the other dense analogue of degeneracy, symmetric difference, is co-NP-complete.
Again, the associate parameterized problem is para-co-NP-complete.

\begin{theorem}\label{thm:sd-npc}
  Deciding if a~graph has symmetric difference at~most~8 is co-NP-complete.
\end{theorem}

This is curious because sd-degeneracy and symmetric difference similarly extend to the dense world two equivalent definitions of degeneracy.
Nevertheless, one can explain the apparent tension between \cref{thm:sdd-npc,thm:sd-npc}: a~vertex ordering witnesses an upper bound in the sd-degeneracy, whereas an induced subgraph witnesses a lower bound in the symmetric difference.
We leave as an open question whether classes of bounded symmetric difference have labeling schemes of (poly)logarithmic size.
This is excluded for bounded sd-degeneracy, for which we now know essentially optimal labeling schemes.
While not an absolute barrier, the likely absence of polynomial certificates tightly upper bounding the symmetric difference complicates matters in settling this open question.

Beside compact labeling schemes, we mentioned \emph{graph coloring} as another motivation for classes of bounded degeneracy.
A~natural dense analogue for bounded chromatic number is the notion of $\chi$-boundedness, that is, the property of having chromatic number bounded by a~function of the clique number.
However, even classes of bounded symmetric difference may not be $\chi$-bounded, as witnessed by shift graphs on pairs.
These graphs have vertex sets $\{(i,j)~:~i < j \in [n]\}$ (for increasing values of~$n$) and $(i,j)$ is adjacent to $(k,\ell)$ whenever $j=k$ or $i = \ell$.
This defines a~class of triangle-free graphs with unbounded chromatic number~\cite{Erdos68}, hence not $\chi$-bounded.
We claim that shift graphs have symmetric difference at~most~2.
Let $H$ be any induced subgraph of a~shift graph.
Consider the smallest integer $i$ such that there two vertices $(a,i)$ and $(b,i)$ in $H$ (or symmetrically the largest~$i$ such that $(i,a), (i,b) \in V(H)$).
One can see that $(a,i)$ and $(b,i)$ (or symmetrically $(i,a), (i,b)$) can have at most one private neighbor each.
If no such integer $i$ exists, $H$ is a~disjoint union of paths, which has symmetric difference at~most~1. 

\textbf{Organization.}
\cref{sec:prelim} gives definitions and notation. In~\cref{sec:stm}, we introduce signed tree models, and prove that graphs of bounded sd-degeneracy admit signed tree models of bounded width.
In~\cref{sec:balance}, we show how to balance these signed tree models, and complete the proof of~\cref{thm:als-sdd}. 
In~\cref{sec:sd-hardness}, we prove~\cref{thm:sd-npc}, and in~\cref{sec:sdd-hardness}, \cref{thm:sdd-npc}.

\medskip
\textbf{Acknowledgments.}
That classes of bounded symmetric difference need not be $\chi$-bounded was observed by Romain Bourneuf, Colin Geniet, Stéphan Thomassé, and the first author.
We thank them for letting us include this remark in the introduction. 

\section{Preliminaries}\label{sec:prelim}

We denote by $[i,j]$ the set of integers that are at~least $i$ and at~most $j$, and $[i]$ is a shorthand for~$[1,i]$.
We follow standard asymptotic notation throughout, and additionally by $f(n)=\tilde{O}(g(n))$ we mean that there exist constants $c,n_0>0$ such that for any $n \geq n_0$ we have $f(n) \leq g(n) \log^c n$. 

We denote by $V(G)$ and $E(G)$ the vertex set and edge set of a~graph~$G$, respectively.
Given a~vertex~$u$ of a~graph $G$, we denote by $N_G(u)$ the set of neighbors of $u$ in $G$ (\emph{open neighborhood}) and by $N_G[u]$ the set $N_G(u) \cup \{u\}$ (\emph{closed neighborhood}).
For a set $S \subseteq V(G)$, we denote by $N_G(S)$ the set of vertices in $V(G) \setminus S$ that have a neighbor in $S$ in graph $G$.
When $H, G$ are two graphs, we may denote by $H \subseteq_i G$ (resp.~$H \subseteq G$) the fact that $H$ is an induced subgraph (resp.~subgraph) of~$G$, i.e., can be obtained by removing vertices of~$G$ (resp.~by removing vertices and edges of~$G$).
We denote by $G[S]$ the subgraph of $G$ induced by $S$, formed by removing every vertex of $V(G) \setminus S$.
We use $G-S$ as a~shorthand for $G[V(G) \setminus S]$, and $G-v$, for $G-\{v\}$.
 
Given two sets $A$ and $B$, we denote by $A \triangle B$ their symmetric difference, that is, $(A \setminus B) \cup (B \setminus A)$.
Given a~graph $G$, and two distinct vertices $u, v \in V(G)$, we set \[\sd_G(u,v) := |(N_G(u) \setminus \{v\}) \triangle (N_G(v) \setminus \{u\})|.\]
The \emph{symmetric difference} of $G$, $\sd(G)$, is defined as $\max_{H \subseteq_i G} \min_{u \neq v \in V(H)} \sd_H(u,v)$.
Symmetric difference was implicitly introduced in \cite{ACLZ15} and later explicitly defined in \cite{AAL21}.
We call \emph{sd-degeneracy} of~$G$, denoted by $\sdd(G)$, the smallest non-negative integer~$d$ such that $|V(G)|=1$ or there is a~pair $u \neq v \in V(G)$ satisfying $\sd_G(u,v) \leqslant d$ and $G-v$ has sd-degeneracy at~most~$d$.
We say that an ordering $v_1, v_2, \ldots, v_n$ of the vertices of $G$ witnesses that the \emph{sd-degeneracy} of~$G$ is at~most~$d$ if for every $i \in [n-1]$, there is a~$j>i$ such that $\sd_{G - \{v_k~:~k \in [i-1]\}}(v_i,v_j) \leqslant d$. 
It thus holds that for any graph $G$, $\sdd(G) \leqslant \sd(G)$, since for every $i \in [n]$, $G - \{v_k~:~k \in [i-1]\}$ is an induced subgraph of~$G$.
But, as shown by \cref{prop:anygraphinSD1}, there are some graphs with \mbox{sd-degeneracy}~1 and unbounded symmetric difference.

Two vertices $u,v$ are said to be \emph{$d$-twins} in a~graph $G$ if they are distinct and $|(N_G(u) \setminus N_G[v]) \cup (N_G(v) \setminus N_G[u])| \leqslant d$.
The $a \times b$ \emph{rook graph} has vertex set $\{(i,j)~:~i \in [a], j \in [b]\}$ and edge set $\{(i,j)(k,\ell)~:~(i,j) \neq (k,\ell),~i=k~\text{or}~j=\ell\}$.
Equivalently it is the line graph of the bipartite complete graph $K_{a,b}$.
For every $a, b \geqslant 3$, the symmetric difference of the $a \times b$ \emph{rook graph} is $2(\min(a,b)-1)$.

We will extensively use \emph{tree orders}, i.e., partial orders defined by ancestor--descendant relationships in a~rooted tree.
We denote by $\prec_T$ the corresponding relation in rooted tree~$T$.
That is, \emph{$u \prec_T u'$} means that $u$ is a~\emph{strict ancestor} of $u'$ in $T$, and \emph{$u \preceq_T u'$} means that $u$ is an \emph{ancestor} of $u'$, i.e., $u=u'$ or $u \prec_T u'$.
We extend this partial order to elements of ${V(T) \choose 2}$.
An unordered pair $uv$ is an \emph{ancestor} of $u'v'$ in $T$, denoted by $uv \preceq_T u'v'$, whenever either $u \preceq_T u'$ and $v \preceq_T v'$, or $v \preceq_T u'$ and $u \preceq_T v'$ holds.
We write $uv \prec_T u'v'$ when $uv \preceq_T u'v'$ and $\{u,v\} \neq \{u',v'\}$.
A~rooted binary tree is \emph{full} if all its \emph{internal nodes}, i.e., non-leaf nodes, have exactly two children.
A~rooted binary tree is \emph{complete} if all its levels are completely filled, except possibly the last one, wherein leaves are left-aligned.
The \emph{depth} of a~rooted tree is the maximum number of nodes in a~root-to-leaf path.

\section{Signed tree models}\label{sec:stm}

A~wide range of structural graph invariants, called width parameters, can be expressed via so-called \emph{tree layouts} (or at~least parameters functionally equivalent to them can).
A~\emph{tree layout} of an \mbox{$n$-vertex} graph $G$ is a~full binary tree $T$ such that the leaves of $T$, that we may denote by $L(T)$, are in one-to-one correspondence with $V(G)$.
Width parameters are typically defined through evaluating a~particular function on bipartitions of $V(G)$ made by the two connected components of $T$ when removing one edge of~$T$.
The width is then the minimum over tree layouts of the maximum over all such evaluations. 
We depart from this viewpoint, and instead augment $T$ with a~sparse structure encoding the graph~$G$.

An unordered pair of vertices in $T$ that is not in an ancestor--descendant relationship is called a \emph{transversal pair of $T$}. 
Two transversal pairs $uv, u'v'$ of $T$ \emph{cross} if 
either $u \prec_T u'$ and $v' \prec_T v$, or $u' \prec_T u$ and $v \prec_T v'$,
or $u \prec_T v'$ and $u' \prec_T v$, or $v' \prec_T u$ and $v \prec_T u'$.
A~\emph{signed tree model} $\mathcal T$ is a~triple $(T,A(T),B(T))$, where $T$ is a~full binary tree, $A(T)$ (for Android green, or Anti) is a~set of transversal pairs of $T$, called \emph{transversal anti-edges}, and $B(T)$ (for Blue, or Biclique) is a~set of transversal pairs of $T$, called \emph{transversal edges}, such that $A(T) \cap B(T) \neq \emptyset$ and no $uv, u'v' \in A(T) \cup B(T)$ cross.
We may refer to the transversal anti-edges as \emph{green edges}, and to the transversal edges as \emph{blue edges}.

\begin{figure}[!ht]
  \centering
   \begin{minipage}[b]{0.47\linewidth}
	  \begin{tikzpicture}[%
	    scale=1.2,
	    level distance=9.5mm,
	    level 1/.style={sibling distance=32mm},
	    level 2/.style={sibling distance=16mm},
	    level 3/.style={sibling distance=8mm},
	    level 4/.style={sibling distance=4mm},
	    every node/.style={draw, circle, inner sep=2pt, font=\footnotesize},
	    edge from parent path={(\tikzparentnode) -- (\tikzchildnode)}
	  ]
	  \node (root) {}
	    child {node (a) {}
	      child {node (b) {}
	        child {node (c) {}
	          child {node (c1) [label=below:{1}] {}}
	          child {node (c2) [label=below:{2}] {}}
	        }
	        child {node (d) [label=below:{3}] {}}
	      }
	      child {node (e) {}
	        child {node (f) {}
	          child {node (f1) [label=below:{4}] {}}
	          child {node (f2) [label=below:{5}] {}}
	        }
	        child {node (g) {}
	          child {node (g1) [label=below:{6}] {}}
	          child {node (g2) [label=below:{7}] {}}
	        }
	      }
	    }
	    child {node (h) {}
	      child {node (i) {}
	        child {node (j) {}
	          child {node (j1) [label=below:{8}] {}}
	          child {node (j2) [label=below:{9}] {}}
	        }
	        child {node (k) [label=below:{10}] {}}
	      }
	      child {node (l) {}
	        child {node (m) {}
	          child {node (m1) [label=below:{11}] {}}
	          child {node (m2) [label=below:{12}] {}}
	        }
	        child {node (n) {}
	          child {node (n1) [label=below:{13}] {}}
	          child {node (n2) [label=below:{14}] {}}
	        }
	      }
	    };
	
	    % transversal anti-edges
	    \foreach \i/\j/\b in {a/n/15, c2/k/0, i/e/0, c2/d/0}{
	      \draw[very thick, agreen] (\i) to [bend left=\b] (\j);
	    }
	
	    % transversal edges
	    \foreach \i/\j/\b in {a/h/0, e/k/0, f/j/28, g1/g2/0, m1/m2/0, m1/n/0, f1/g/0, n/i/0, c/d/0}{
	      \draw[very thick, blue] (\i) to [bend left=\b] (\j);
	    }
		\end{tikzpicture}
	  	\caption{A~signed tree model of a~14-vertex graph. 
	  }
		\label{fig:signed-tree-model}
	\end{minipage}
	\hfill
	\begin{minipage}[b]{0.47\linewidth}
		\begin{tikzpicture}[%
		scale=1.2,
		level distance=9.5mm,
		level 1/.style={sibling distance=32mm},
		level 2/.style={sibling distance=16mm},
		level 3/.style={sibling distance=8mm},
		level 4/.style={sibling distance=4mm},
		every node/.style={draw, circle, inner sep=2pt, font=\footnotesize},
		edge from parent path={(\tikzparentnode) -- (\tikzchildnode)}
		]
		\node (root) {}
		child {node (a) {}
			child {node (b) {}
				child {node (c) {}
					child {node (c1) [label=below:{1}] {}}
					child {node (c2) [label=below:{2}] {}}
				}
				child {node (d) [label=below:{3}] {}}
			}
			child {node (e) {}
				child {node (f) {}
					child {node (f1) [label=below:{4}] {}}
					child {node (f2) [label=below:{5}] {}}
				}
				child {node (g) {}
					child {node (g1) [label=below:{6}] {}}
					child {node (g2) [label=below:{7}] {}}
				}
			}
		}
		child {node (h) {}
			child {node (i) {}
				child {node (j) {}
					child {node (j1) [label=below:{8}] {}}
					child {node (j2) [label=below:{9}] {}}
				}
				child {node (k) [label=below:{10}] {}}
			}
			child {node (l) {}
				child {node (m) {}
					child {node (m1) [label=below:{11}] {}}
					child {node (m2) [label=below:{12}] {}}
				}
				child {node (n) {}
					child {node (n1) [label=below:{13}] {}}
					child {node (n2) [label=below:{14}] {}}
				}
			}
		};
		
		% transversal anti-edges
		\foreach \i/\j/\b in {a/n/15, c2/k/0, i/e/0, c2/d/0, b/e/0, c1/c2/0, f1/f2/0, f/g/0, j1/j2/0, j/k/0, i/l/0, m/n/0, n1/n2/0}{
			\draw[very thick, agreen] (\i) to [bend left=\b] (\j);
		}
		
		% transversal edges
		\foreach \i/\j/\b in {a/h/0, e/k/0, f/j/28, g1/g2/0, m1/m2/0, m1/n/0, f1/g/0, n/i/0, c/d/0}{
			\draw[very thick, blue] (\i) to [bend left=\b] (\j);
		}
		\end{tikzpicture}
		\caption{The signed tree model of~\cref{fig:signed-tree-model} made clean.}
		\label{fig:clean-signed-tree-model}		
	\end{minipage}
\end{figure}

The \emph{width} of the signed tree model $(T,A(T),B(T))$ is the degeneracy of the graph $(V(T),A(T) \cup B(T))$.
Note that if $(V(T),A(T) \cup B(T))$ is $d$-degenerate, then $(V(T),A(T) \cup B(T) \cup E(T))$ is \mbox{$(d+2)$-degenerate}.
The signed tree model is \emph{$d$-sparse} if $|A(T) \cup B(T)| \leqslant d |V(T)|$.
We observe that a~signed tree model of width $d$ is $d$-sparse, but an~$O(1)$-sparse signed tree model can have width $\Omega(\sqrt{|V(T)|})$.

The signed tree model $\mathcal T := (T,A(T),B(T))$ defines a~graph $G := G_{\mathcal T}$ with vertex set $L(T)$.
Two leaves $u, v \in L(T)$ are adjacent in $G$ if there is $u'v' \in B(T)$ such that 
%$u' \preceq_T u$ and $v' \preceq_T v$, 
$u'v' \preceq_T uv$,
and there is no $u''v'' \in A(T)$ with $u'v' \prec_T u''v'' \preceq_T uv$.
For example, in the representation of~\cref{fig:signed-tree-model}, vertices 4 and 8 are adjacent in $G$ because of the blue edge between their parents (below the green edge between their grandparents), but vertices 7 and 8 are non-adjacent because of the green edge between their grandparents (below the blue edge between their great-grand-parents).
We may say that a~graph~$G$ admits (or has) a~signed tree model of width~$d$ if there is a signed tree model of this width that defines~$G$.
Every graph~$G$ admits a~signed tree model as one can simply set $A(T) := \emptyset,\,B(T):=E(G)$ on an arbitrary full binary tree $T$ with $L(T)=V(G)$.
However this representation may have large width, while a~more subtle one (linking nodes higher up in the tree) may have a lower width.

Every graph of twin-width~$d$ admits a~signed tree model with $A(T) = \emptyset$ and width at~most~$d+1$.
\emph{Tree models} or \emph{twin-decompositions} are signed tree models with $A(T) = \emptyset$, and further technical requirements.
We observe that similar objects to signed tree models were utilized in~\cite{twin-width5} to attain a~fast matrix multiplication on matrices of low twin-width.
We will not need a~definition of twin-width, and refer the interested reader to~\cite{twin-width1}.
In~\cref{sec:intro} we also mentioned  Welzl orders with low alternation number~\cite{Welzl88}, let us now elaborate on that.

A~\emph{Welzl order} of \emph{alternation} \emph{number}~$d$ for a~graph $G$ is a~total order $<$ on $V(G)$ such that the neighborhood of every vertex is the union of at~most $d$~intervals along $<$.
We claim that bipartite graphs $G=(X \uplus Y, E(G))$ with a~Welzl order $<$ of alternation number~$d$ admit a~signed tree model of width $2d$.
Note that we can assume that for every $x \in X$ and $y \in Y$, $x<y$.
We build a~signed tree model $(T,A(T),B(T))$ of $G$ as follows.
Let us call~\emph{left binary comb} a~full binary tree whose internal nodes induce a~path, rooted at an endpoint of this path, and every right child is a~leaf.
We make the root of $T$ adjacent to the roots of two left binary combs with $|X|$ and $|Y|$ leaves, respectively.
The leaves are labeled from left to right with the vertices of $G$ in the order~$<$.
To simplify the notations, assume that these labels describe $[n]$ in the natural order.
To represent that vertex $i \in X$ has $[j,k] \subseteq Y$ in the partition of its open neighborhood into maximal intervals, we add a~blue edge between leaf~$i$ and the parent of~$k$, and a~green edge between~$i$ and the sibling of~$j$ (to stop the interval).
Finally observe that $(V(T),A(T) \cup B(T))$ has degeneracy at~most~$2d$.
(The subtree whose leaves are the vertices of $X$ need not be a~binary comb.)

A~similar construction would work for graphs $G$ of chromatic number $q$, and would yield a~signed tree model of width~$2(q-1)d$.
A~more permissive definition of signed tree models, allowing leaf-to-ancestor transversal edges and relaxing the notion of crossing, would give models of width~$2d$ for any graph with a~Welzl order of alternation number~$d$.
However, with this alternative definition, the consequences of the next section would not follow.
Hence we stick to the given definition of signed tree models.

A~signed tree model is said to be \emph{clean} if every pair of siblings are linked by a~green or blue edge.
It is easy to turn a~signed tree model into a~clean one representing the same graph: simply add green edges between every pair of siblings that were previously not linked (by a blue or green edge).
This operation may only increase the width of the signed tree model by~1. 
The advantage of working with a clean signed tree model is that for every pair of leaves $u, v$ with least common ancestor~$w$, there is at~least one transversal edge or anti-edge connecting the paths (in $T$) between $w$ and $u$ and between $w$ and $v$.
Clean tree models will be useful in~\cref{sec:balance} when we balance the trees associated with the tree models. 

Given a clean signed tree model~$(T,A(T),B(T))$ and $u, v \in L(T)$, we denote by \emph{$e_T(u,v)$} the unique green or blue edge $u'v'$ such that $u'v' \preceq_T uv$ and no green or blue edge $u''v''$ satisfies $u'v' \prec_T u''v'' \preceq_T uv$.
The edge $e_T(u,v)$ exists because the signed tree model is clean, and is unique because no green or blue edges may cross (or be equal).
Then, $u,v$ are adjacent in $G$ if and only if $e_T(u,v)\in B(T)$, i.e., $e_T(u,v)$ is a~blue edge.
We first show that graphs of bounded sd-degeneracy (and in particular, of bounded symmetric difference) admit clean signed tree models of bounded width.

\begin{lemma}\label{lem:sd-degen-to-stm}
  Any graph of sd-degeneracy~$d$ admits a~clean signed tree model of width~$d+1$.
\end{lemma}
\begin{proof}
  Let $v_1, \ldots, v_n$ be a~vertex ordering that witnesses~sd-degeneracy~$d$ for an $n$-vertex graph~$G$.
  For $i \in [n]$, let $G_i := G - \{v_j~:~1 \leqslant j \leqslant i-1\}$.
  In particular, $G_1 = G$.
  Let $u_i$ be a~$d$-twin of~$v_i$ in~$G_i$.
  Initially we consider a~forest of $n$ distinct $1$-vertex rooted trees, each root labeled by a~distinct vertex of~$G$.
  We will build $T$ (and in parallel, the transversal anti-edges and edges) by iteratively giving a~common parent to two roots of this forest of $n$ singletons.
  Note that different nodes of~$T$ may have the same label, as the labels will range in $V(G)$ whereas $T$ has $2n-1$ nodes.
  
  For $i$ ranging from 1 to $n-1$:
  \begin{compactitem}
  \item  add a~blue (resp.~green) edge between $v_i$ and $u_i$ if $u_iv_i \in E(G)$ (resp.~$u_iv_i \notin E(G)$),
  \item add a~blue edge between $v_i$ and the roots labeled by $w$ for $w \in N_{G_i}(v_i) \setminus N_{G_i}[u_i]$,
  \item add a~green edge between $v_i$ and the roots labeled by $w$ for $w \in N_{G_i}(u_i) \setminus N_{G_i}[v_i]$, and
  \item create a~common parent, labeled by $u_i$, for the roots labeled $u_i$ (left child) and $v_i$ (right child).
  \end{compactitem}
  This defines a~full binary tree~$T$ such that $L(T)=V(G)$.
  In $(V(T),A(T) \cup B(T))$, the leaves labeled by $v_1$ and $u_1$ have degree at most $d+1$ and 1, respectively.
  Hence an immediate induction on $(T,A(T),B(T))$ (after removing these two leaves, and following the order $v_2, \ldots, v_n$) shows that $(V(T),A(T) \cup B(T))$ is $(d+1)$-degenerate.
  As we only add transversal anti-edges and edges between pairs of roots, no pair in $A(T) \cup B(T)$ can cross.
  Indeed, if $x, y$ are two nodes of~$T$ that are both roots in some $G_i$, then it cannot happen that $x',y'$ are also both roots of some $G_{i'}$ with $x \prec_T x'$ and $y' \prec_T y$.
  The first item further ensures that the signed tree model $(T,A(T),B(T))$ of width~$d+1$ is clean.
  
  Let us finally check that for every $u, v \in L(T)$, $e_T(u,v)$ is a~blue edge if and only if $uv \in E(G)$.
  This is a~consequence of the following property.
  \begin{claim}\label{clm:lowering-transv-edges}
    Let $x, y$ be two nodes of $T$ labeled by $u, v$ respectively.
    Let $x'$ be a~child of $x$, labeled by $u'$, such that $x'y$ is neither a~blue nor a~green edge.
    Further assume that $y$ was a~root when the parent of $x'$ (i.e., $x$) was created.
    Then, $uv \in E(G)$ if and only if $u'v \in E(G)$.
  \end{claim}
  \begin{poc}
    If $x'$ is the left child of $x$, the conclusion holds since $u = u'$.
    We can thus assume that $x'$ is the right child of $x$, and not the sibling of~$y$ since it would contradict that $x'y$ is neither a~blue nor a~green edge.
    Node $x'$ was not linked to $y$ by a~blue or a~green edge, so $v$ cannot be a~neighbor of exactly one of $u, u'$.
  \end{poc}
  Consider the moment $e_T(u,v)$ was added to the signed tree model, say between the then-roots $x$ and $y$, labeled by $u'$ and $v'$, respectively. 
  By the way blue and green edges are introduced, $xy$ is a~blue edge if $u'v' \in E(G)$, and $xy$ is green if $u'v' \notin E(G)$.
  Thus we conclude by iteratively applying~\cref{clm:lowering-transv-edges}.
\end{proof}

\section{Balancing Signed Tree Models}\label{sec:balance}

For any signed tree model of width $d$ of an $n$-vertex graph, we get an adjacency labeling scheme with labels of size $O(d h \log n)$, where $h$ is the depth of~$T$. 
Indeed, one can label a~leaf $v$ of $T$ (i.e., vertex of~$G$) by the identifiers (each of $\log(2n)$ bits) of all the nodes of the path from $v$ to the root of~$T$, followed by the identifiers of the outneighbors of these at~most~$h$ nodes in a~fixed orientation of $(V(T),A(T) \cup B(T))$ with maximum outdegree at~most~$d$, allocating an extra bit for the color of each corresponding edge.
One can then decode the adjacency of any pair $u, v \in V(G)$ by looking at the color of $e_T(u,v)$.
The latter is easy to single out, based on the labels of~$u$ and~$v$.

\begin{proposition}\label{prop:als}
  Let $G$ be an $n$-vertex graph with a~signed tree model of width~$d$ and depth~$h$.
  Then, $G$ admits an~$O(d h \log n)$-bit adjacency labeling scheme.
\end{proposition}

Unfortunately, the depth of the tree $T$ of a~signed tree model of low width obtained for an \mbox{$n$-vertex} graph of low sd-degeneracy could be as large as~$n$.
This makes a~direct application of~\cref{prop:als} inadequate.
Instead, we first decrease the depth of the signed tree model, while controlling its sparsity. 
We start with some notation and auxiliary tools. Throughout this section we fix $n$ and denote by $R$ a full, complete, rooted binary tree whose leaves are natural numbers $1, 2, \ldots, n$ read from left to right.

\begin{lemma}\label{lem:interval-union}
 Any interval $[i,j]$ with $i, j \in [n]$ is the disjoint union of the leaves of at~most~$2 \log n$ rooted subtrees of~$R$. 
\end{lemma}
\begin{proof}
  Let $X \subseteq V(R)$ be such that the leaves of the subtrees rooted at a~node of $X$ partition $[i,j]$, and $X$ is of minimum cardinality among node subsets with this property.
  Let $k$ be the first level of~$R$ intersected by~$X$ (with the root being at level~$1$).
  At~most two nodes $x, y$ of~$X$ are at level~$k$ (and exactly one node when $k=2$), with $x=y$ or $x$ to the left of $y$.
  Observe that if $x \neq y$, then $x, y$ have to be consecutive along the left-to-right ordering of level~$k$, but cannot be siblings (otherwise they can be substituted by their parent).
  At level $k+1$, at~most two nodes can be part of~$X$: the node just to the left of the leftmost child of~$x$, and the node just to the right of the rightmost child of~$y$.
  This property propagates to the last level.
  Thus $|X| \leqslant \max(2 (\lceil \log n \rceil - 1), \lceil \log n \rceil) \leqslant 2 \log n$.
\end{proof}

From the previous proof it can also be seen that there is a~unique minimum-cardinality set~$X$ representing an interval $I = [i,j]$, which we denote by $X_I$. Furthermore, we observe that the minimality of $X_I$ implies that the elements of $X_I$ are incomparable in $R$, i.e., for any distinct $x,y \in X_I$, neither $x \prec_R y$, nor $y \prec_R x$.

Let $T$ be any full binary tree with $n$ leaves that are natural numbers  $1, 2, \ldots, n$ read from left to right. For any node $x$ of $T$ the leaves of the subtree of~$T$ rooted at~$x$ form an interval in $[n]$, which we denote by~$I_T(x)$.
We proceed with further auxiliary statements.

\begin{observation}\label{obs:tree-of-intervals}
  For every $x, y \in V(T)$, if $I_T(x)$ and $I_T(y)$ intersect, then one is included in the other, and thus $x \preceq_T y$ or $y \preceq_T x$.
\end{observation}

The observation above helps us establish the following lemma. 

\begin{lemma}\label{obs:pushed-up}
	Let $p$ and $s$ be two distinct nodes of $T$, and $a$ and $b$ be two nodes of $R$ such that $a \in X_{I_T(p)}$ and $b \in X_{I_T(s)}$.
	If $a \preceq_R b$, then $p \preceq_T s$ or $s \preceq_T p$.
	Furthermore, if  $a \prec_R b$, then $p \prec_T s$. 
\end{lemma}
\begin{proof}
	Since $a \in X_{I_T(p)}$ and $b \in X_{I_T(s)}$, we have $I_R(a) \subseteq I_T(p)$ and 
	$I_R(b) \subseteq I_T(s)$. Thus, if $a \preceq_R b$, then $I_R(b) \subseteq I_R(a)$, and therefore $I_T(p)$ and $I_T(s)$ intersect. By \cref{obs:tree-of-intervals}, this implies that $I_T(p) \subseteq I_T(s)$ or $I_T(s) \subseteq I_T(p)$, and thus $p \preceq_T s$ or $s \preceq_T p$.

	Suppose now that $a \prec_R b$.
        Since the elements of $X_{I_T(s)}$ are incomparable in $R$ and $b$ belongs to this set, we conclude that $a \not\in X_{I_T(s)}$.
        Thus, again by minimality of $X_{I_T(s)}$, there exists a leaf $x \in I_R(a)$ that does not belong to $I_T(s)$, and since $I_R(a) \subseteq I_T(p)$, we conclude that $I_T(p) \not\subseteq I_T(s)$.
        Therefore, we must have $I_T(s) \subsetneq I_T(p)$, where the strictness of the inclusion is due to the fact that $x \in I_T(p) \setminus I_T(s)$.
        Consequently, we conclude that $p \prec_T s$, as desired.
\end{proof}  

We are now ready to prove the main lemma of this section. Recall that $R$ a full, complete, rooted binary tree with $n$ leaves, and $T$ is a full binary tree with $n$ leaves, where in both trees the leaves are natural numbers $1, 2, \ldots, n$ read from left to right.

\begin{lemma}\label{lem:shallowisation}
	Let $(T,A(T),B(T))$ be a~clean $d$-sparse signed tree model of an $n$-vertex~graph $G$.
	Then, $G$ admits a~$4d \log^2 n$-sparse signed tree model $(R,A(R),B(R))$ of depth $\lceil \log n \rceil + 1$. 
\end{lemma}
\begin{proof}
	We start by describing the construction of $A(R) \cup B(R)$.
	For every transversal anti-edge (resp.~edge) $xy \in A(T)$ (resp.~$xy \in B(T)$), we add to $A(R)$ (resp.~$B(R)$) all the unordered pairs $ab$ with $a \in X_{I_T(x)}$ and $b \in X_{I_T(y)}$; we say that each such pair $ab$ \emph{originated} from $xy$.
	It may happen that some $ab$ is added to both~$A(R)$ and~$B(R)$.
	In which case, $ab$ originated from both $x_0y_0 \in A(T)$ and $x_1y_1 \in B(T)$ such that $x_0y_0 \prec_T x_1y_1$ or $x_1y_1 \prec_T x_0y_0$.
	In the former case, we remove $ab$ from $A(R)$ (and only keep it in $B(R)$), and in the latter, we remove $ab$ from $B(R)$ (and only keep it in $A(R)$).
	This finishes the construction of $\mathcal R := (R,A(R),B(R))$.
	
	Let us first argue that no two transversal pairs in $A(R) \cup B(R)$ cross.
	Assume for the sake of contradiction that $a_1b_1, a_2b_2  \in A(R) \cup B(R)$ satisfy $a_1 \prec_{R} a_2$ and $b_2 \prec_{R} b_1$.
	Let $p_1s_1, p_2s_2 \in A(T) \cup B(T)$ be transversal pairs from which $a_1b_1, a_2b_2$, respectively, originated. By \cref{obs:pushed-up}, we have $p_1 \prec_{T} p_2$ and $s_2 \prec_{T} s_1$, implying that $p_1s_1$ and $p_2s_2$ cross in $T$, which is not possible. This contradiction proves that $\mathcal R$ is a~signed tree model.

	Next, let us show that $\mathcal R$ represents~$G$.
	Fix $u, v \in V(G)$ and let $xy := e_T(u,v)$. Without loss of generality, assume that $xy$ belongs to $A(T)$.
	Let $ab$ be the transversal pair originating from $xy$ such that $u \in I_{R}(a)$ and $v \in I_{R}(b)$. By construction $ab$ belongs to $A(R) \cup B(R)$.
	We claim that $ab \in A(R)$ and $ab = e_R(u,v)$.
	First, suppose that $ab \not\in A(R)$, and thus $ab \in B(R)$. By construction, this is possible only if $ab$ has also originated from some $x_0y_0 \in B(T)$ and $xy \prec_T x_0y_0$. Note that as $u \in I_{R}(a)$ and $v \in I_{R}(b)$, we have that $a \preceq_R u$ and $b \preceq_R v$. Thus, by \cref{obs:pushed-up}, $x_0 \preceq_T u$ or $u \preceq_T x_0$, and $y_0 \preceq_T v$ or $v \preceq_T y_0$, and since $u$ and $v$ are leaves in $T$, we conclude that $x_0y_0 \preceq _T uv$. Consequently, $xy \prec_T x_0y_0 \preceq_T uv$, which contradicts the fact that $xy = e_T(u,v)$.
	Now, suppose that $ab \neq e_R(u,v)$, i.e., there exists $a_1b_1$ such that 
	$ab \prec_R a_1b_1 \preceq_R uv$. Again by \cref{obs:pushed-up}, this implies
	$xy \prec_T x_1y_1 \preceq_T uv$, where $x_1y_1$ is the transversal pair of $T$ from
	which $a_1b_1$ originated. This contradicts the assumption that  $xy = e_T(u,v)$.
		
	Finally, we establish the claimed sparseness of $\mathcal R$.
	By design, the depth of $R$ is $\lceil \log n \rceil + 1$.
	As $\mathcal T := (T,A(T),B(T))$ is~$d$-sparse, it has at~most $(2n-1)d$ transversal (anti-)edges.
	Each blue or green edge of $\mathcal T$ gives rise to at~most $(2 \log n)^2$ blue or green edges of $\mathcal T'$, by~\cref{lem:interval-union}.
	Hence $\mathcal R$ is $4d \log^2 n$-sparse.
\end{proof}

We finally need this folklore observation.
\begin{observation}\label{obs:nbr-edges-to-degen}
  Every $m$-edge graph has degeneracy at~most~$\lceil \sqrt{2m} \rceil - 1$.
\end{observation}
\begin{proof}
  It is enough to show that any $m$-edge graph $G$ has a~vertex of degree at~most $\lceil \sqrt{2m} \rceil - 1$.
  If all the vertices of~$G$ have degree at~least $\lceil \sqrt{2m} \rceil$, then $m \geqslant \frac{1}{2} n \lceil \sqrt{2m} \rceil$.
  But also $n \geqslant \lceil \sqrt{2m} \rceil+1$ for a~vertex to possibly have $\lceil \sqrt{2m} \rceil$ neighbors.
  Thus $m \geqslant \frac{1}{2} \sqrt{2m} (\sqrt{2m}+1)> m$, a~contradiction.
\end{proof}

Combining~\cref{lem:sd-degen-to-stm,lem:shallowisation,obs:nbr-edges-to-degen,prop:als} yields~\cref{thm:als-sdd}.
\begin{proof}[Proof of~\cref{thm:als-sdd}]
  Let $G$ be an $n$-vertex graph of sd-degeneracy~$d$.
  By~\cref{lem:sd-degen-to-stm}, $G$ admits a~clean signed tree model of width at~most~$d+1$, hence $(d+1)$-sparse.
  Thus by~\cref{lem:shallowisation}, $G$ has a~$4(d+1)\log^2 n$-sparse signed tree model $\mathcal T$ of depth $\lceil \log n \rceil + 1$.
  By~\cref{obs:nbr-edges-to-degen}, $\mathcal T$ has width at~most \[\sqrt{16(d+1)n \log^2 n}=4 \sqrt{(d+1)n} \log n.\]
  Therefore, by \cref{prop:als}, $G$ has a~$O(\sqrt{dn}\log^3n)$-bit labeling scheme. 
\end{proof}

\section{Symmetric Difference is para-co-NP-complete}\label{sec:sd-hardness}
  
For any fixed even integer $d \geqslant 8$, we show that the following problem is NP-complete:
Does the input graph $G$ have an induced subgraph with at least two vertices and no pair of $d$-twins?
We call such an induced subgraph a~\emph{$(d+1)$-diverse graph}.
The membership of this problem to NP is straightforward, as a~$(d+1)$-diverse induced subgraph $H$ of~$G$ is a~polynomial-size witness.
One can indeed check in polynomial-time that $H$ has at least two vertices, and that for every pair $u,v$ of vertices of $H$, at least $d+1$ other vertices of $H$ are neighbors of exactly one of~$u,v$.

The \emph{$d$-twin graph} $T_d(G)$ of a~graph $G$ is a~graph with vertex set $V(G)$ and edges between every pair of $d$-twins.
\begin{observation}\label{obs:sol-in-IS}
  The vertices of a~$(d+1)$-diverse induced subgraph of $G$ form an independent set of~$T_d(G)$.
\end{observation}

Given any \textsc{3-SAT} formula $\varphi$ with at~most~three occurrences of each variable, clauses of size two or three, and at~least three clauses, we build a~graph $G := G(\varphi)$ such that $G$ has a~$(d+1)$-diverse induced subgraph if and only if $\varphi$ is satisfiable.
Such a~restriction of \textsc{3-SAT} is known to be NP-complete~\cite{Tovey84}.

\subsection{Bubble gadget}

A~\emph{bubble gadget} $B$ (or \textit{bubble} for short) is a $w \times w$ rook graph, with $w := \frac{d}{2}+2$, deprived of the two rightmost vertices of its top row.
Let the gadget $B$ be an induced subgraph of a graph $G$ and let $S$ be the set of neighbors in $G$ of the bubble outside of~$B$.
We say that $B$ is \emph{properly attached} to~$S$ in~$G$ if each vertex of the top row (of width $\frac{d}{2}$) and of the rightmost column (of height $\frac{d}{2}+1$) has one or two neighbors outside the gadget, whereas the other vertices of~$B$ have no neighbors outside~$V(B)$.
%Let $S$ be the set of neighbors of the bubble outside of~$B$.

We say that $B$ is \emph{neatly attached to $S$} if it is properly attached to $S$, and further, vertices of the top row and rightmost column have \emph{exactly} one outside neighbor, and at most one vertex of~$S$ has neighbors in both the top row and rightmost column.
The \emph{neat} attachments that we will use all satisfy $3 \leqslant |S| \leqslant 5$.
Hence they can be described by a~tuple of size between 3 and~5, listing the number of neighbors of vertices in $S$ among $V(B)$, starting with the top row and ending with the rightmost column.
For instance,~\cref{fig:gen-bubble-gadget} depicts a~neat $(2,2,2,7)$-attachment.
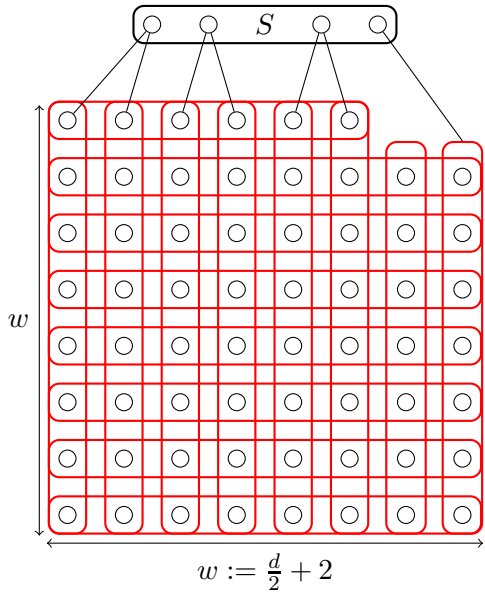
\begin{figure}[!ht]
  \centering
  \begin{tikzpicture}[vertex/.style={draw,circle,inner sep=0.08cm}]
    \def\s{0.75}
    \def\w{7}
    \def\h{8}
    \def\p{4}
    \pgfmathtruncatemacro\wm{\w-1}
    \pgfmathtruncatemacro\wp{\w+1}

    \foreach \i/\j in {1/2.5,2/3.5,3/5.5,4/6.5}{
      \node[vertex] (u\i) at (\j * \s,0.7 * \s) {} ;
    }
    \node[draw,thick,rounded corners,label={[yshift=-0.52cm]$S$},fit=(u1) (u4)] (s) {} ;
    
    \foreach \i in {1,...,\wp}{
      \foreach \j in {2,...,\h}{
        \node[vertex] (v\i\j) at (\s * \i, - \s * \j) {} ;
      }
    }
    \foreach \i [count = \im from 1] in {2,...,\w}{
        \node[vertex] (v\im1) at (\s * \im, - \s) {} ;
    }

    \node[draw,red,thick,rounded corners,fit=(v11) (v\wm1)] (c1) {} ;
    \foreach \j in {\w,\wp}{
      \node (z\j) at (\s * \j, - \s * 2 + 0.2) {} ;
      \node[draw,red,thick,rounded corners,fit=(z\j) (v\j\h)] (d\j) {} ;
    }

    \foreach \i [count = \im from 1] in {2,...,\w}{
      \node[draw,red,thick,rounded corners,fit=(v\im1) (v\im\h)] (c\im) {} ;
    }

    \foreach \i/\j in {1/1,1/2,2/3,2/4,3/5,3/6}{
        \draw (u\i) -- (v\j1) ;
    }
    \draw (u4) -- (d\wp.north) ; 

    \foreach \j in {2,...,\h}{
       \node[draw,red,thick,rounded corners,fit=(v1\j) (v\wp\j)] (d\j) {} ;
    }

    \draw[<->] (\s/2 + 0.1,- \h * \s - 0.5 * \s) to node[midway,below] {$w := \frac{d}{2} + 2$} (\wp * \s + \s/2 - 0.1,- \h * \s - 0.5 * \s);
    \draw[<->] (0.5 * \s,-\s + 0.2) to node[midway,left] {$w$} (0.5 * \s,-\h * \s - \s/2 +0.1);
  \end{tikzpicture}
  \caption{A~neatly $(2,2,2,7)$-attached bubble gadget, with $d=12$.
  The vertical and horizontal red boxes are cliques.}
  \label{fig:gen-bubble-gadget}
\end{figure}
A~bubble properly attached to $S$ is in a~delicate state.
It may entirely survive in a~$(d+1)$-diverse induced subgraph of~$G$.

\begin{observation}\label{obs:bubble-stable}
  Let $B$ be a~bubble gadget properly attached to~$S$ in~$G$.
  No pair of vertices of $B$ are $d$-twins in $G[V(B) \cup S]$. 
\end{observation}
\begin{proof}
  In $B$ the only pairs with symmetric difference at~most~$d$, in fact exactly~$d$, consist of a~vertex in the top row and another vertex in its column, or two vertices of the same row in the two rightmost columns.
  In both cases, these pairs have symmetric difference at~least~$d+1$ in $G[V(B) \cup S]$ since the vertices of the top row or rightmost column have at~least one neighbor in $S$, while all other vertices of $B$ have no neighbor in~$S$.
\end{proof}

However, deletions that cause one vertex of the top row or two vertices of the rightmost column to no longer have outside neighbors cause the bubble to completely collapse.

\begin{lemma}\label{lem:deflators}
  Let $B$ be a~bubble gadget properly attached to~$S$ in~$G$.
  Let $H$ be any $(d+1)$-diverse induced subgraph of~$G$, such that at least one vertex of the top row or at least two vertices of the rightmost column has no neighbor in $V(H) \setminus V(B)$. Then, $H$ contains at~most one vertex of~$B$. 
\end{lemma}
\begin{proof}
  We first deal with the case when a~vertex $v$ of the top row (in the entire $B$) has no neighbor in $V(H) \setminus V(B)$.
  By symmetry, assume that $v$ is the topmost vertex of the first column.
  Vertex $v$ is thus $d$-twin with all the other vertices of the first column.
  Hence by~\cref{obs:sol-in-IS}, either $v$ is not in $H$, or none of the $d/2+1$ vertices below~$v$ are in $H$.

  If the latter holds, then any two vertices in the same column, outside the top row and rightmost column, are now $d$-twins in $H$.
  By~\cref{obs:sol-in-IS}, within these vertices, $H$ can only contain at most one vertex per column.
  In turn, the kept vertices are pairwise $d$-twins in $H$, so at most one can be kept overall.
  We conclude since the vertices of $N_H(S) \cap V(B)$ have at most~two neighbors in~$S$.

  We now suppose that $v$ is not in~$H$.
  Then, in each row but the topmost, the vertices in the first and penultimate columns are $d$-twins in $H$.
  Thus, within each pair, at most one vertex can be in $H$.
  This implies that any two vertices in the same column, outside the top row and rightmost column, are now $d$-twins in $H$.
  Thus we conclude as in the previous paragraph.

  We now deal with the case when two vertices $x,y$ of the right most column have no neighbor in $V(H) \setminus V(B)$.
  By symmetry, we can assume that $x$ is in the second row, and $y$ is in the third row.
  Then, in $H$, $x$ (resp.~$y$) is $d$-twin with the vertex just to its left.
  After one vertex is removed in each pair, in each column but the last two, the vertices in the second and third rows have become $d$-twins in $H$.
  Therefore, $H$ can only contain at most one vertex from all these pairs.
  We reach the state that any two vertices in the same row, outside the top row and rightmost column, are $d$-twins in $H$, and we can conclude as previously.
\end{proof}

In the incoming construction, all the bubble gadgets will be neatly attached.
Furthermore, every vertex a~bubble is attached to will have at least one neighbor on the top row, or at least two neighbors in the rightmost column.
Thus the deletion of \emph{any} vertex a~bubble~$B$ is attached to will result, by~\cref{lem:deflators}, in deleting all the vertices of~$B$ but at~most~one.

\subsection{Variable and clause gadgets}

The \emph{variable gadget} of variable $x$ used in $\varphi$ is simply two vertices $x, \neg x$ adjacent to a~set $N_x$ of $t := \frac{d}{2}+1$ shared neighbors.
\begin{figure}[!ht]
  \centering
  \begin{tikzpicture}[vertex/.style={draw,circle,inner sep=0.08cm}]
    \def\s{1}
    \def\t{7}

    \node at (-1.35 * \s, 0) {$x$} ;
    \node[vertex] (x) at (-\s,0) {} ;
    \node[vertex] (nx) at (\s,0) {} ;
    \node at (1.45 * \s, 0) {$\neg x$} ;
    
    \foreach \i in {1,...,\t}{
      \node[vertex] (a\i) at (\i * \s - \t * \s * 0.5 - \s * 0.5, -\s) {} ;
    }
    \node at (\s - \t * \s * 0.5 - \s * 1.1, -\s) {$N_x$} ;
    
    \node[draw,thick,rounded corners,fit=(a1) (a\t)] (T) {} ;

    \foreach \i in {1,...,\t}{
      \draw (x) -- (a\i) -- (nx) ;
    }
   
  \end{tikzpicture}
  \caption{The variable gadget of $x$ with $d=12$.}
  \label{fig:variable-gadget}
\end{figure}
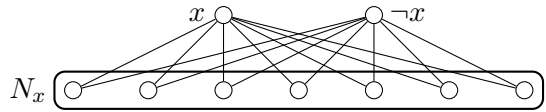
Vertices $x$ and $\neg x$ have one or two other neighbors in $G$ corresponding to the clause they belong to, as we will soon see.
Some vertices of the first two sets $N_x$ ($N_{x_1}$ and $N_{x_2}$) will have neighbors in these two sets.
Apart from $N_{x_1}$, each set $N_x$ will form an independent set.
All other neighbors of the vertices of any $N_x$ will belong to at~most four bubble gadgets.
This will be described in detail shortly.
The \emph{clause gadget} of clause $c$ consists of a~pair of adjacent vertices $v_c, d_c$.
We make $v_c$ (but not $d_c$) adjacent to the two or three vertices corresponding to the literals of~$c$.

\subsection{Construction of $G(\varphi)$}\label{subsec:sd-construction}

Unsurprisingly, we add one variable gadget per variable, and one clause gadget per clause of $\varphi$. 
Let $x_1, \ldots, x_n$ be a~numbering of the variables, and $c_1, \ldots, c_m$, of the clauses.
We neatly attach a~bubble gadget to $S_j$ made of the five vertices $z_j,v_{c_j},d_{c_j},$ $v_{c_{j+1}},d_{c_{j+1}}$ for every $j \in [m-1]$, with (in this order) a~$(1, \lfloor d/4 \rfloor, \lfloor d/4 \rfloor, \lceil d/4 \rceil, \lceil d/4 \rceil)$-attachment, where $z_j$ is a~vertex of some $N_x$.
The choice of $z_j$ is irrelevant, but we take all the vertices $z_j$ pairwise distinct.
This is possible since there are at most $3n/2$ clauses, and more than $dn/2$ vertices contained in the union of the sets~$N_x$.
For every $j \in [m-1]$, we make $\{v_{c_j}, d_{c_j}\}$ fully adjacent to $\{v_{c_{j+1}}, d_{c_{j+1}}\}$.
The construction of $G$ is almost complete; see \cref{fig:overall-clause-gadgets} for an illustration.

\begin{figure}[!ht]
  \centering
  \resizebox{\textwidth}{!}{
  \begin{tikzpicture}[vertex/.style={draw,circle,inner sep=0.08cm}]

    % clause gadgets
    \begin{scope}[xshift=3.6cm]
    \def\s{1.65}
    \def\m{6}
    \pgfmathtruncatemacro\mm{\m-1}

    \node at (\s-0.5,0) {$v_{c_1}$} ;
    \node at (6.33 * \s+0.57,0) {$d_{c_6}$} ;
    
    \foreach \i in {1,...,\m}{
      \node[vertex] (c\i) at (\i * \s, 0) {} ;
      \node[vertex] (d\i) at (\i * \s + \s/3, 0) {} ;
      \draw (c\i) -- (d\i) ;
    }

    \foreach \i [count = \ip from 2] in {1,...,\mm}{
      \draw (c\i) to [bend left = 20] (c\ip) ;
      \draw (d\i) -- (c\ip) ;
      \draw (c\i) to [bend left = 25] (d\ip) ;
      \draw (d\i) to [bend left = 20] (d\ip) ;
    }
    
    \foreach \i [count = \ip from 2] in {1,...,\mm}{
      \node[draw,blue,thick,ellipse,inner sep=0.3cm,fit=(c\ip) (d\i)] (e\i) {} ;
    }
    \end{scope}

    % variable gadgets
    \begin{scope}
    \def\s{0.4}
    \def\n{5}
    \def\t{7}

    \foreach \j in {1,...,\n}{
      \begin{scope}[yshift=-2.6cm, xshift=3.2 * \j cm]
      \node at (-2 * \s, 0) {$x_\j$} ;
      \node[vertex] (x\j) at (-\s,0) {} ;
      \node[vertex] (nx\j) at (\s,0) {} ;
      \node at (2.1 * \s, 0) {$\neg x_\j$} ;
    
      \foreach \i in {1,...,\t}{
        \node[vertex] (a\j\i) at (\i * \s - \t * \s * 0.5 - \s * 0.5, -1.3* \s) {} ;
      }
      %\node at (\s - \t * \s * 0.5 - \s * 1.1, -\s) {$N_{x_\j}$} ;
    
    \node[draw,thick,rounded corners,fit=(a\j1) (a\j\t)] {} ;

    \foreach \i in {1,...,\t}{
      \draw (x\j) -- (a\j\i) -- (nx\j) ;
    }
    \end{scope}
    }
  \end{scope}

  % variable-clause incidence
    \foreach \i/\j in {1/x1,1/nx3,1/x4, 2/nx2,2/x3,2/nx4, 3/nx1,3/x2,3/nx5, 4/nx4,4/x5, 5/x1,5/x3, 6/x2,6/nx5}{
      \draw (c\i) -- (\j) ;
    }
    
  \end{tikzpicture}
  }
  \caption{The essential part of~$G$ built so far, for a~3-CNF formula $\varphi$ whose first two clauses are $x_1 \lor \neg x_3 \lor x_4$ and $\neg x_2 \lor x_3 \lor \neg x_4$.
  The blue ellipses represent the bubbles attached to the four enclosed vertices (recall that the bubble is attached to a~fifth vertex among the sets $N_x$).}
  \label{fig:overall-clause-gadgets}
\end{figure}
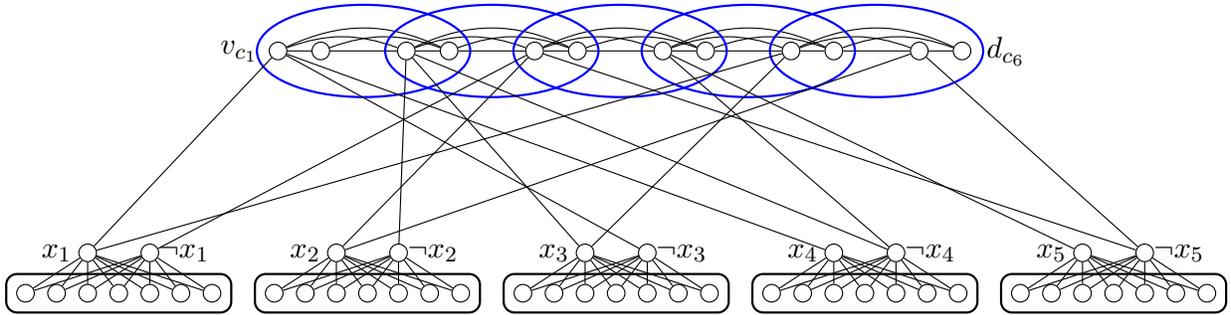

At this point, all the vertices $v_{c_j},d_{c_j}$ such that $j \in [2,m-1]$ have exactly $\lfloor d/4 \rfloor + \lceil d/4 \rceil = d/2$ neighbors in (two) bubble gadgets.
Let $y_1, \ldots, y_{nt}$ be the ordering of the vertices in $\bigcup_x N_x$ from left to right in how they appear in~\cref{fig:overall-clause-gadgets}.
We neatly attach 
\begin{enumerate}
	\item a~bubble gadget to $(v_{c_1}, d_{c_1}, y_1)$ by a~$(\lceil d/4 \rceil, \lceil d/4 \rceil, d+1-2 \lceil d/4 \rceil)$-attachment;
	\item a~bubble gadget to $(v_{c_m},d_{c_m}, y_{nt})$ by a~$(\lfloor d/4 \rfloor, \lfloor d/4 \rfloor, d+1-2 \lfloor d/4 \rfloor)$-attachment;
	\item for every $i \in [nt-2]$, a~bubble gadget to $S'_i$ made of the three vertices $y_i, y_{i+1}, y_{i+2}$ with a~$(1,d/2,d/2)$-attachment.
\end{enumerate}
Finally, we make a clique on the first $\lceil d/4 \rceil$ vertices of $N_{x_1}$ (i.e., on the vertices $y_1, y_2, \ldots, y_{\lceil d/4 \rceil}$) and fully connect them to the first $\lceil d/4 \rceil$ vertices of $N_{x_2}$ (i.e., to the vertices $y_{t+1}, y_{t+2}, \ldots, y_{t+\lceil d/4 \rceil}$).
This finishes the construction.

We make some observations.
As all the bubble gadgets are neatly attached, no two vertices outside a~bubble gadget $B$ can share a~neighbor in~$B$.
\begin{observation}\label{obs:d/2}
  For every $j \in [m]$, $v_{c_j},d_{c_j}$ each have exactly $d/2$ neighbors in bubble gadgets (all of which are non-adjacent to any other vertex outside their respective bubble).
\end{observation}

Vertices in $\bigcup_x N_x$ have more neighbors in bubbles.
\begin{observation}\label{obs:d}
  For every $i \in [nt]$, $y_i$ has at least $d/2+1$ neighbors in bubble gadgets (all of which are non-adjacent to any other vertex outside the respective bubble).
  Furthermore, if $i \geq 3$, then $y_i$ has at least $d$ neighbors in bubble gadgets.
\end{observation}

\subsection{Correctness}\label{subsec:sd-correctness}

We can now show the following strengthening of~\cref{thm:sd-npc}.
\begin{theorem}\label{thm:sd-npc-gen}
  For every fixed even integer $d \geqslant 8$, deciding if an input graph has symmetric difference at~most~$d$ is co-NP-complete.
\end{theorem}

As the graph $G:=G(\varphi)$ presented in~\cref{subsec:sd-construction} can be constructed in polynomial time, to prove \cref{thm:sd-npc-gen}, we shall simply check the equivalence between the satisfiability of $\varphi$ and the existence of a $(d+1)$-diverse induced subgraph of $G(\varphi)$.
We do this in the next two lemmas.
Recall that, by definition, $G$ does \emph{not} have symmetric difference at~most~$d$ if and only if it has a~$(d+1)$-diverse induced subgraph.

\begin{lemma}\label{lem:intended-solution}
  If $\varphi$ is satisfiable, then $G$ admits a~$(d+1)$-diverse induced subgraph.
\end{lemma}
\begin{proof}
  Let $\mathcal A$ be a~satisfying assignment of~$\varphi$.
  For each variable $x$ of $\varphi$, we delete vertex $\neg x$ if $\mathcal A$ sets $x$ to true, and we delete vertex $x$ otherwise (if $\mathcal A$ sets $x$ to false).
  Let us call $H$ the obtained induced subgraph of $G$ (with at least two vertices).
  We claim that $H$ has no pair of $d$-twins, and successively rule out such pairs
  \begin{compactenum}[$(i)$]
  \item within the same bubble,
  \item between a~vertex in a~bubble $B$ and a vertex outside $B$ (but possibly in another bubble), 
  \item between two vertices both outside every bubble gadget.
  \end{compactenum}
  \medskip
  $(i)$ As $H$ contains all the vertices of $G$ on which bubble gadgets are attached, by \cref{obs:bubble-stable}, no two distinct vertices in the same bubble are $d$-twins.

  \medskip
  \noindent
  $(ii)$ Let us fix a~bubble gadget~$B$ attached to $S$, and two vertices $u \in V(B)$ and $v \in V(H) \setminus V(B)$.
  First observe that $u$ has at least $d/2+1$ neighbors in $V(B)$ (hence in $H$) that are not neighbors of~$v$.
  All the vertices $v \in V(H) \setminus (V(B) \cup S)$ have at least $d/2$ neighbors in $H$ that are not neighbors of~$u$.
  For these vertices $v$, $\sd_H(u,v) > d$.
  We thus focus on the case when $v \in S$.
  Recall that bubble gadgets are attached only to vertices in $\{ y_1, \ldots, y_{nt} \} \cup \bigcup_{j \in [m]} \{v_{c_j}, d_{c_j}\}$. First, assume $v \in \{ y_3, y_4, \ldots, y_{nt} \}$. Each such vertex has at~least $d/2$ neighbors in each of two distinct bubbles, and thus it has $d/2$ neighbors outside~$B$ implying $\sd_H(u,v) > d$. Next, assume $v \in \{y_1,y_2\}$. Then $v$ has one neighbor in some bubble gadget, say $B_1$, and at least $d/2$ neighbors in some other gadget, say $B_2$. If $B = B_1$, then clearly $\sd_H(u,v) > d$. If $B=B_2$, then $v$ has at least $(\lceil d/4 \rceil-2) + \lceil d/4 \rceil + 1 + 1 \geq d/2$ neighbors outside $B$ that are not neighbors of $u$, where $\lceil d/4 \rceil-2$ neighbors are coming from the clique $\{ y_1, y_2, \ldots, y_{\lceil d/4 \rceil} \}$ (note that one of those vertices is $v$, and at most one of them is a neighbor of $u$, as $u$ can have at most one neighbor outside $B$, hence minus 2), $\lceil d/4 \rceil$ are coming from $\{ y_{t+1}, y_{t+2}, \ldots, y_{t+\lceil d/4 \rceil} \}$, one neighbor is coming from $\{x_1, \neg x_1\}$, and one more is coming from $B_1$. Thus, in this case, we also have $\sd_H(u,v) > d$.

  Finally, assume that~$v$ is $v_{c_j}$ or $d_{c_j}$ for some $j \in [m]$.
  We can further assume that~$u$ is in the top row or rightmost column of~$B$, otherwise it has $d$ neighbors that are not neighbors of~$v$ (and~$v$ has at least one private neighbor).
  Now we observe that
  \[\sd_H(u,v) \geqslant |N_H(u) \setminus N_H[v]|+|N_H(v) \setminus N_H[u]| \geqslant d/2+1 + d/2-1 - \lceil d/4 \rceil + \lfloor d/4 \rfloor + 2 \geqslant d+1,\]
  where $d/2+1$ lower bounds the number of neighbors of~$u$ whose neighborhood is included in $V(B)$, $d/2-1 - \lceil d/4 \rceil$ lower bounds the number of neighbors of~$u$ in the top row or rightmost column of~$B$ that are not adjacent to~$v$, $\lfloor d/4 \rfloor$ lower bounds the number of neighbors of~$v$ in another bubble than~$B$, and~2 accounts for the 
  at~least two neighbors of $v$ that are not neighbors of $u$ among $\{v_{c_{j-1}},d_{c_{j-1}}, v_{c_{j}},d_{c_{j}}, v_{c_{j+1}},d_{c_{j+1}}\}$, whichever exist.
  (Here we use the fact that $\varphi$ has at least two clauses.)

  \medskip
  \noindent
  $(iii)$ Let $u, v$ be two distinct vertices outside every bubble gadget.
  First, recall that, by~\cref{obs:d/2,obs:d}, every vertex $v_{c_j},d_{c_j}, j \in [m]$ has $d/2$ neighbors in bubble gadgets, and every vertex in $\bigcup_x N_x$ has at least $d/2+1$ neighbors in bubble gadgets.
  Since no two vertices outside bubble gadgets can share a~neighbor in a~bubble gadget, we conclude that if $u \in \bigcup_x N_x \cup \bigcup_{j \in [m]} \{v_{c_j}, d_{c_j}\}$ and $v \in \bigcup_x N_x$, then $u$ and $v$ are not $d$-twins.
  
  Furthermore, if $u$ is $x_i$ or $\neg x_i$ for some $i \in [n]$ and $v \in \{ y_3, y_4, \ldots, y_{nt} \}$, then, by \cref{obs:d} and the fact that $u$ has at least one neighbor $v_{c_j}$ (for some $j \in [m]$), the vertices $u$ and $v$ are also not $d$-twins.   
  Next, since each $x_i$ and $\neg x_i$ has $d/2+1$ neighbors in $N_{x_i}$, and no neighbors in $N_{x_{i'}}$ for $i' \neq i$, we conclude that if $u$ is $x_i$ or $\neg x_i$ for some $i \in [n]$ and $v$ is $x_{i'}$ or $\neg x_{i'}$ for some $i' \in [n]$, or $v_{c_j}$ or $d_{c_j}$ for some $j \in [m]$, then $u$ and $v$ are not $d$-twins.
  
  Now, assume that $u$ is $x_i$ or $\neg x_i$ for some $i \in [n]$, and $v \in \{y_1,y_2\}$.
  Note that none of $y_1$ and $y_2$ has neighbors in $N_{x_i}$ for $i \geq 3$, and therefore a~non-trivial argument is only required for $i \leq 2$.
  Then $v$ has at least $\lceil d/4 \rceil - 1$ neighbors in $N_{x_1} \cup N_{x_2}$ that are not neighbors of~$u$, and by \cref{obs:d}, $v$ has at least $d/2+1$ additional private neighbors in bubble gadgets.
  On the other hand, $u$ has at least $d/2 + 1 - \lceil d/4 \rceil$ private neighbors in $N_{x_2}$.
  This implies that $u$ and $v$ are not $d$-twins in $H$.
   
   This leaves us with the cases when $u$ and $v$ are two vertices in clause gadgets.
   By \cref{obs:d/2}, each of these vertices has $d/2$ private neighbors in bubble gadgets.  
   As there are at least three clauses in $\varphi$, two vertices $u,v$ from distinct clause gadgets have at~least~two additional private neighbors.
   Thus, we can assume that $u=v_{c_j}$ and $v=d_{c_j}$ for some $j \in [m]$.
   As $\mathcal A$ is a~satisfying assignment, at~least one vertex $x$ or $\neg x$ adjacent to $v_{c_j}$ has survived in~$H$.
   Hence $\sd_H(u,v) \geqslant d/2+d/2+1=d+1$.
\end{proof}

\begin{lemma}\label{lem:no-diverse-ind-sub}
  If $\varphi$ is not satisfiable, then $G$ has no $(d+1)$-diverse induced subgraph.
\end{lemma}
\begin{proof}
	 Suppose, toward a~contradiction, that $\varphi$ is not satisfiable, but there exists a set $S \subseteq V(G)$ such that $G[S]$ is a $(d+1)$-diverse graph.
         Our strategy is to start with $V' = V(G)$ and iteratively reduce $V'$ while ensuring that it still contains $S$.
         To do this, we will repeatedly apply~\cref{lem:deflators} and the fact that at most one vertex in any pair of $d$-twins in $G[V']$ can belong to $S$.
         Eventually, we will reduce $V'$ to a single vertex, and thus arrive at a~contradiction that $G[S]$ is a~$(d+1)$-diverse graph. 
	
	For each variable $x$, the vertices $x, \neg x$ are 3-twins, thus at least one of them has to be removed in a~$(d+1)$-diverse induced subgraph.
	The kept literals (if any) define a~(partial) truth assignment.
	By assumption, this assignment does not satisfy at~least one clause $c_j$.
	This implies that $v_{c_j}, d_{c_j}$ are $d$-twins in the corresponding induced subgraph.
	Indeed, they each have exactly $d/2$ private neighbors in bubble gadgets, and no other private neighbors.
	
	By~\cref{lem:deflators}, the bubble attached to $S_j$ is reduced to at~most one vertex, say $w_j$ (if any).
	In turn, this makes the pairs $v_{c_{j-1}}, d_{c_{j-1}}$ and $v_{c_{j+1}}, d_{c_{j+1}}$ $d$-twins (when they exist).
	Indeed their symmetric difference is at~most $3+2\lceil d/4 \rceil +1\leqslant d$ (as $d$ is even and at~least~8), where 3 accounts for the three literals of the clause, and 1 for vertex $w_j$.
	This iteratively collapses every bubble attached to some $S_{j'}$ to a~single vertex, as well as the two bubble gadgets attached to $\{v_{c_1},d_{c_1}, y_1\}$ and $\{v_{c_m},d_{c_m},y_{nt}\}$, to say, $w_0$ and $w_m$.
	Now all the vertices $w_j$ (for $j \in [0,m]$) have degree at~most~1, hence are pairwise 2-twins.
	So at~most one of them can be kept.
	We recall that at~most one vertex per clause gadget could be kept.
	For $j$ going from 1 to $m-1$, the vertex kept (if any) from the clause gadget of $c_j$ is an 8-twin of the vertex kept in the next surviving clause gadget.
	This implies that from all the clause gadgets and all the bubble gadgets attached to them, one can only keep at~most one vertex overall, say $z$.
	This vertex has degree at~most~3 in the resulting induced subgraph.
	
	Vertices $y_1, y_2$ are now $d$-twins, so the bubble gadget attached to $S'_1$ collapses to at~most one vertex.
	Vertices $y_1$ and $z$ are now $d$-twins, so at~most one can survive, which we keep calling~$z$.
	Even if $y_2$ is kept, it is now a~$d$-twin of $y_3$, thus at~most one of $y_2, y_3$ can be kept. 
	This implies the collapse of the bubble gadget attached to $S'_2$ to at~most one vertex, absorbed by~$z$.
	In turn, $y_2$ and $z$ collapse to a~single vertex.
	This process progressively eats up all the vertices $y_j$, and all the bubble gadgets attached to them. 
	As soon as a~vertex $x$ or $\neg x$ has three remaining neighbors, it becomes a~6-twin of~$z$, and is absorbed by~it.
	We end up with the single vertex~$z$.
\end{proof}

\newcommand{\out}{\text{out}}
\newcommand{\bubtree}{\mathcal{B}}
 
\section{Deciding sd-degeneracy at~most~1 is NP-complete}\label{sec:sdd-hardness}

\newcommand{\Var}{\text{Var}}
\newcommand{\Clause}{\text{Clause}}
\renewcommand{\top}{\uparrow}
\renewcommand{\bot}{\downarrow}
\newcommand{\isolated}{\iota}
\newcommand{\final}{\gamma}

The membership to NP of computing the sd-degeneracy is readily given by the witnessing order, which is a~polynomial certificate.
Indeed, for any integer $d$, testing if a~vertex has a $d$-twin can be done in polynomial time.
In this section, we reduce \textsc{3-SAT} where every variable occurs in two or three clauses and every clause has three literals to deciding if the sd-degeneracy of a~graph is at~most~1. As mentioned in the previous section, such a~restriction of \textsc{3-SAT} is known to be NP-complete~\cite{Tovey84}. More precisely, the reduction is from determining if a~variable assignment satisfies all but at~most~one clause.
Obviously this problem remains NP-hard: simply duplicate an instance of \textsc{3-SAT} on two disjoint sets of variables.
One can satisfy all but at~most~one clause of the resulting instance if and only if the original instance is satisfiable.

For a given formula~$\varphi$, we define a graph $G := G(\varphi)$, which is of sd-degeneracy at~most~1 if and only if all clauses of~$\varphi$ but at~most one are satisfied by some assignment.
The graph~$G$ is made of variable gadgets, clause gadgets, and two extra vertices: a~vertex of large degree~$\final$ and an isolated vertex~$\isolated$.

\subsection{Construction of $G(\varphi)$}

For each variable~$v$, we define the \emph{variable gadget}~$\Var(v)$.
Let~$p$ be the number of clauses~$v$ occurs in (note that~$p$ is either 2 or 3).
The gadget $\Var(v)$ is made of $2p + 1$ vertices.
Two special vertices~$v_0$ and~$v_1$ named respectively \emph{representatives of}~$\neg  v$ and~$v$,
and~$2p - 1$ vertices named the \emph{transition vertices from~$v_0$ to~$v_1$}.
The~$2p + 1$ vertices of~$\Var(v)$ form an independent set.
For a literal $l$ of variable $v$, we will write $\Var(l)$ to refer to gadget $\Var(v)$.

For each clause~$c$ appearing in~$\varphi$, we define the \emph{clause gadget}~$\Clause(c)$.
Let~$l_1, l_2$ and~$l_3$ be the literals of the clause~$c$.
The gadget~$\Clause(c)$ is made of~five vertices:~$c_{\top}, c_{l_1}, c_{l_2}, c_{l_3}$ and~$c_{\bot}$.
The five vertices of~$\Clause(c)$ form a clique.
The graph~$G$ is obtained from the disjoint union of the variable gadgets, clause gadgets, and two extra vertices~$\final$ and~$\isolated$, on which we add the following edges.

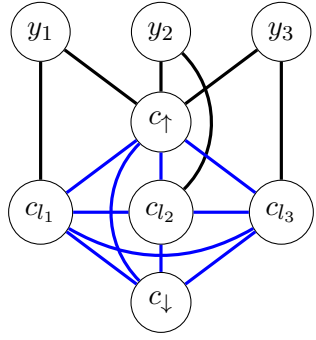
\begin{figure}[h!]
\centering
\begin{tikzpicture}[vertex/.style={circle, draw, minimum size=0.4cm}]
\def\s{0.8}

\foreach \i/\j/\l in {0/3/y_1, 2/3/y_2, 4/3/y_3, 0/0/c_{l_1}, 2/0/c_{l_2}, 4/0/c_{l_3}} {
	\node[vertex] (\l) at (\i * \s,\j * \s) {$\l$};
}
\node[vertex](ctop) at (2 * \s, 1.5 * \s) {$c_\top$};
\node[vertex](cbot) at (2 * \s, -1.5 * \s) {$c_\bot$};
\foreach \i/\j in {y_1/c_{l_1}, y_3/c_{l_3}, y_1/ctop, y_2/ctop, y_3/ctop} {
	\draw[very thick] (\i) -- (\j);
}
\draw[very thick] (y_2) to[out=-45,in=45] (c_{l_2});
\foreach \i/\j in {ctop/c_{l_1}, ctop/c_{l_3}, ctop/c_{l_2}, cbot/c_{l_1}, cbot/c_{l_3}, cbot/c_{l_2}, c_{l_2}/c_{l_1}, c_{l_2}/c_{l_3}} {
	\draw[very thick, blue] (\i) -- (\j);
}
\draw[very thick, blue] (c_{l_1}) to [bend left=-30] (c_{l_3});
\draw[very thick, blue] (ctop) to[out=-135,in=135] (cbot);
\end{tikzpicture}
\caption{Adjacencies between the clause gadget of $c = l_1 \lor l_2 \lor l_3$ and $y_1, y_2$ and $y_3$ the representatives of $l_1, l_2$, and $l_3$, respectively.
The blue edges represent the internal edges of the clause gadget.}
\label{fig:var-clause connexion}
\end{figure}

\begin{compactenum}
	\item \label{cli ctop edges} For each clause~$c = l_1 \lor l_2 \lor l_3$ and $i \in [3]$, we link both~$c_{l_i}$ and~$c_{\top}$ to the representative of~$l_i$; see~\cref{fig:var-clause connexion}.
	\item \label{cbot edges} For each clause~$c$, we add an edge between~$c_{\bot}$ and~$\final$. 
	\item \label{transition edges}
At this point, for each variable gadget~$\Var(v)$, one can describe the neighborhood of the representatives $v_0$ and $v_1$ using two integers $a$ and $b$, where $a+b$ is the number of occurrences of variable $v$ in $\varphi$.
Recall from \cref{cli ctop edges} that the neighborhood of~$v_0$ is a set of $2a$ vertices in clause gadgets (two vertices per clause containing $\neg v$), say~$\{x_1, \ldots, x_{2a}\}$, where we order them such that the vertices~$x_i$ for $i \in [a]$ are of the form~$c_\top$ for some clause~$c$, and the vertices $x_i$ for $i \in [a+1,2a]$ are of the form $c_{l_j}$.
In the neighborhood of~$v_1$, which is a set~$\{y_1, \ldots, y_{2b}\}$ of $2b$ vertices,
the order is reversed. Namely, the vertices~$y_i$ for $i \in [b]$ are of the form $c_{l_j}$, and the vertices~$y_i$ for $i \in [b+1,2b]$ are of the form~$c_\top$ for some clause~$c$.
Let $\{t_1, \dots, t_{2(a+b)-1}\}$ be the transition vertices from $v_0$ to $v_1$.
We set $N(t_1)$ to~$N(v_0) \cup \{y_1\}$, and for~$i \in [2,2b]$, we set $N(t_i)$ to $N(t_{i-1}) \cup \{y_i\}$.
Note that~$N(t_{2b}) = N(v_0) \cup N(v_1)$.
Then for~$i \in [2b+1,2a+2b-1]$, we set $N(t_{i})$ to~$N(t_{i-1}) \setminus \{x_{i - 2b}\}$. 
\end{compactenum}
We finally add $\isolated$, an isolated vertex.
This finishes the construction; see~\cref{fig:global picture} for the overall picture.

\begin{figure}[h!]
\centering
\usetikzlibrary{calc,shapes.arrows,decorations.pathreplacing}
\begin{tikzpicture}[vertex/.style={circle, draw, minimum size=0.45cm}]

\def\varsize{2.3}

\node at (\varsize - 0.5, 5.3) {$\Var(x_1)$} ;

\node (rpz_pos) at (1.5*\varsize, 6) {representative of $x_2$};
\node (rpz_neg) at (4 * \varsize, 6) {representative of $\neg x_2$};
\node (trans_vert) at (2.75*\varsize, 5.5) {transition vertices};

\foreach \v/\name in {0/1, 2/2, 5.2/3} {
	\draw[very thick, rounded corners, blue] (\v * \varsize, 5) rectangle ( \v * \varsize +  1.5*\varsize, 4) {};
	\node[vertex] (pos\name) at (\v * \varsize + 0.2 * \varsize, 4.5) {};
	\node[vertex] (neg\name) at (\v * \varsize + 1.3 * \varsize, 4.5) {};
	\foreach \i\l in {1/1, 2.2/2, 5/3} {
		\node[vertex, minimum size = 0.05] (tv\name\l) at (\v * \varsize + 0.3 * \varsize + 0.15 * \i * \varsize, 4.5) {};
	}
        \node at  (\v * \varsize + 0.3 * \varsize + 0.15 * 3.65 * \varsize, 4.5) {$\ldots$} ;

}
\node at  (10, 4.5) {$\ldots$} ;

\draw[->] (rpz_pos) -- (pos2);
\draw[->] (rpz_neg) -- (neg2);
\foreach \l in {1, 2, 3} {
	\draw[->] (trans_vert) -- (tv2\l);
}

\begin{scope}[yshift=-0.1cm]
\def\clausesize{2.5}

\node at (\varsize - 0.5, -0.4) {$\Clause(c_1)$} ;

\node at  (10.5, 1.5) {$\ldots$} ;

\foreach \c/\j/\o in {0/1/0.04, 2/2/0.04, 5/m/0.02} {
	\draw[very thick, rounded corners, red] (\c * \clausesize, 3) rectangle ( \c * \clausesize +  1.5*\clausesize, 0) {};
	
	\node[vertex,inner sep=\o cm] (ctop\c) at (\c * \clausesize + 0.75 * \clausesize, 2.5) {\tiny{$(c_\j)_\top$}};
        \node[vertex,inner sep=\o cm] (cbot\c) at (\c * \clausesize + 0.75 * \clausesize, 0.5) {\tiny{$(c_\j)_\bot$}};

	\foreach \i/\l in {0.7/1, 2/2, 3.3/3} {
		\ifthenelse{\c  = 0} {
			\node[vertex,inner sep=0.02 cm] (c_lit\c \l) at (\c * \clausesize + 0.37 * \i * \clausesize, 1.5) {\tiny{$(c_\j)_{x_\l}$}};
		}	{
			\ifthenelse{\c = 2} {
                                \node[vertex,inner sep=0.01 cm] (c_lit\c \l) at (\c * \clausesize + 0.37 * \i * \clausesize, 1.5) {\tiny{$(c_\j)_{\neg x_\l}$}};
			}	{
			\node[vertex] (c_lit\c \l) at (\c * \clausesize + 0.37 * \i * \clausesize, 1.5) {};
			}
		}
	}

	\foreach \i/\j in {ctop\c/c_lit\c 1, ctop\c/c_lit\c 2, ctop\c/c_lit\c 3, cbot\c/c_lit\c 1, cbot\c/c_lit\c 2, cbot\c/c_lit\c 3, c_lit\c 1/c_lit\c 2, c_lit\c 2/c_lit\c 3} {
		\draw[very thick, blue] (\i) -- (\j);
	}
	
	\draw[very thick, blue] (c_lit\c 1) to[bend left=-35] (c_lit\c 3);
	\draw[very thick, blue] (ctop\c) to[out=-135,in=135] (cbot\c);

}
\end{scope}

\node[vertex] (gamma) at (4 * \varsize, -1) {$\final$};
\node[vertex] (iso) at (5 * \varsize, -1) {$\isolated$};
\foreach \c in {0, 2, 5} {
	\draw (cbot\c) -- (gamma);
}

\draw[very thick] (pos1) -- (c_lit01);
\draw[very thick] (pos1) -- (ctop0);
\draw (tv11) -- (c_lit01);
\draw (tv11) -- (ctop0);
\draw (tv11) to [bend left=-15] (ctop2);

\draw (tv12) -- (c_lit01);
\draw (tv12) -- (c_lit21);
\draw (tv12) -- (ctop0);
\draw (tv12) -- (ctop2);

\draw[very thick] (neg1) -- (c_lit21);
\draw[very thick] (neg1) -- (ctop2);
\draw[very thick] (neg1) to [bend left=-8] (c_lit51);
\draw[very thick] (neg1) to [bend left=-8] (ctop5);

\end{tikzpicture}
\caption{Depiction of $G(\varphi)$ when the first two clauses of $\varphi$ are $c_1 = x_1 \vee x_2 \vee x_3$ and $c_2 = \neg x_1 \vee \neg x_2 \vee \neg x_3$, literal $x_1$ appears only once, and $\neg x_1$ is the first literal of the second and last clauses, $c_2$ and  $c_m$ respectively.
The blue boxes represent the variable gadgets and the red boxes represent the clause gadgets.
%The blue edges are the edges inside clause gadgets.
For legibility, only few edges between the variable gadgets and the clause gadgets were drawn.}
\label{fig:global picture}
\end{figure}
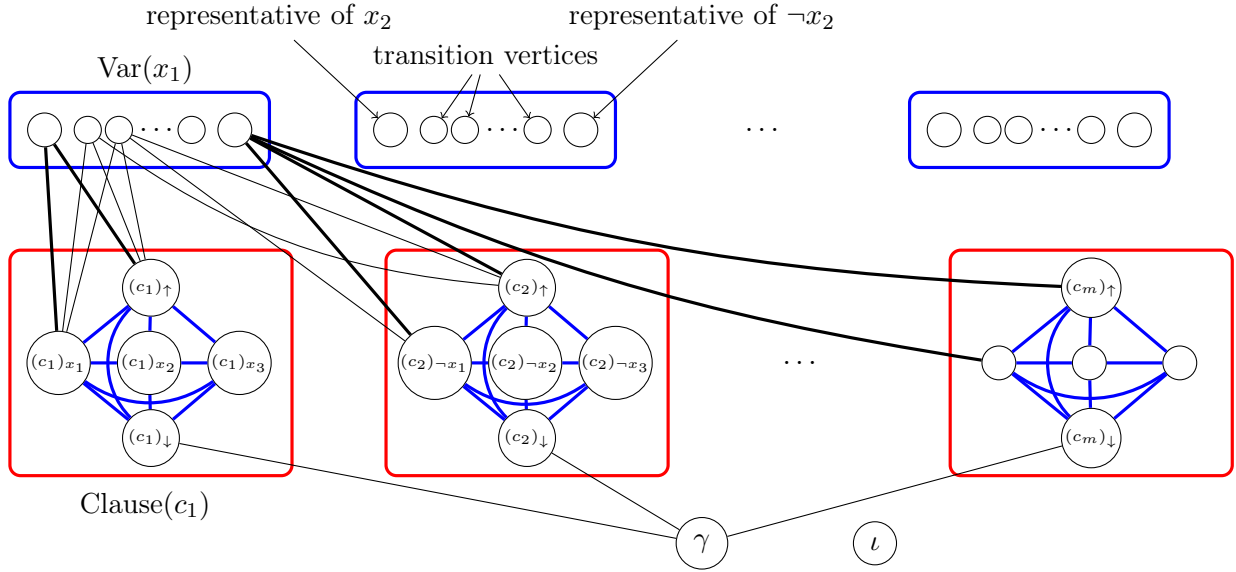

\subsection{Correctness}
Henceforth we say that a~vertex $v$ can be \emph{eliminated} in an induced subgraph $H$ of $G$, if $v$ has a~1-twin in $H$, and an elimination order of~$G$ is an ordering $u_1, u_2, \ldots$ of $V(G)$, equivalently represented by a~total order $\prec$ on $V(G)$ such that $u_i$ can be eliminated in $G - \{u_j~:~1 \leqslant j \leqslant i-1\}$.
Given an elimination order~$\prec$ of~$G$ and a vertex~$v$, we denote the graph $G - \{u \in V(G)~:~ u \prec v\}$ by~$G_v$.

Let us first give an overview of the proof.
In one direction, (\cref{lem:sat-to-order}) we turn an assignment satisfying all but at~most one clause into an elimination order.
This is relatively straightforward. We start by eliminating in each variable gadget $\Var(v)$ every vertex but $v_0$ when $v$ is set to true, and every vertex but $v_1$ when $v$ is set to false.
After which, we can remove every vertex of $\Clause(c)$ for every satisfied clause $c$.
From there one can easily finish the elimination order of~$G$.

Let us now consider the other direction (\cref{lem:order-to-sat}). If~$G$ admits an elimination order~$\prec$, an assignment is obtained by considering the last vertex eliminated in each variable gadget.
If the neighborhood of this last vertex contains~$N(v_0)$, then~$v$ is set to true, otherwise~$v$ is set to false.
We denote this assignment by~$\mathcal A (\prec)$.
We need to argue that $\mathcal A (\prec)$ satisfies all but at~most one clause of~$\varphi$.

First, we show that since~$\final$ is adjacent to all the vertices~$c_\bot$, for $\final$ to be eliminated one needs to have at~least started the elimination of all but at~most one clause gadget.
Second, we prove that for a~clause gadget~$\Clause(c)$ to start its elimination before $\final$ is eliminated, one needs, for at~least one of its literals~$l$, that all the vertices of $\Var(l)$ that are fully adjacent to the neighborhood of its representative are already eliminated.
Finally, we show that for each variable~$v$, the last eliminated vertex of~$\Var(v)$ is eliminated after each clause gadget adjacent to~$\Var(v)$ has started to be eliminated.
This implies that~$\mathcal A(\prec)$ satisfies all but at~most one clause of~$\varphi$.

%In what follows, we denote by $\Var(l)$ the gadget corresponding to the variable contained (positively or negatively) by the literal~$l$.

\begin{lemma}\label{lem:sat-to-order}
If all but at~most one clause of $\varphi$ can be satisfied, then $G$ has sd-degeneracy at~most~1.
\end{lemma}
\begin{proof}
Let $\mathcal A$ be a~variable assignment satisfying all but at~most one clause of $\varphi$; say, $c_0$ if an unsatisfied clause indeed exists.
We give an elimination order of $G$.
For each variable $v$, we eliminate in $\Var(v)$ every vertex but $v_0$ if $v$ is set to true, and every vertex but $v_1$, if instead $v$ is set to false.
This is possible since $v_0, t_1, \ldots, t_h, v_1$ is a~\emph{path} of 1-twins, where $t_1, \ldots, t_h$ are the transition vertices of $\Var(v)$.

Now we eliminate the clause gadgets.
Let $c = l_1 \lor l_2 \lor l_3$ be a~clause satisfied by $\mathcal A$.
Assume without loss of generality (since $\Clause(c)$ is symmetric) that $l_1$ is satisfied by $\mathcal A$.
So the only remaining vertex in $\Var(l_1)$ is the representative of $\neg l_1$.
In particular, $\Var(l_1)$ is no longer adjacent to $\Clause(c)$.
Thus $c_\top$ and $c_{l_2}$ are 1-twins: they may only differ on the last remaining vertex of $\Var(l_3)$.
We eliminate $c_\top$.
In turn, $c_{l_1}$ and $c_{l_2}$ are 1-twins: they may only differ on the last remaining vertex of $\Var(l_2)$.
We eliminate~$c_{l_2}$.
The same holds for $c_{l_1}$ and $c_{l_3}$, and we eliminate $c_{l_3}$.
At this point, $\Clause(c)$ contains exactly two vertices: $c_{l_1}$ and $c_\bot$.
The neighborhood of $c_{l_1}$ is reduced to $\{c_\bot\}$, and the neighborhood of $c_\bot$ is $\{c_{l_1}, \final\}$.
We eliminate $c_{l_1}$ followed by $c_\bot$; the latter is now a 1-twin of $\isolated$, since it has degree~1.

All the vertices of all the clause gadgets except possibly $\Clause(c_0)$ are now eliminated.
Hence~$\final$ is of degree at most~1, thereby is a~1-twin of $\isolated$.
We eliminate $\final$.
It remains to eliminate $\Clause(c_0)$ (if it exists).
Let $l$ be a~literal of $c_0$.
Now $c_\bot$ and $(c_0)_l$ are 1-twins.
We eliminate $(c_0)_l$ for each literal $l$ of $c_0$.
The vertex $(c_0)_\bot$ and all remaining vertices of the variable gadgets are now of degree at~most~1.
We eliminate them, as 1-twin of $\isolated$.
We finally eliminate $(c_0)_\top$, and remain with a~single vertex,~$\isolated$.
\end{proof}

We now proceed with some lemmas which will establish the other direction of the reduction. Given two sets $X, Y \subseteq V(G)$, we write $X \prec Y$ if for every $x \in X$ and $y \in Y$ it holds that $x \prec y$, and we write $x \prec Y$ or $Y \prec x$ when $X=\{x\}$.

\begin{lemma}\label{lem:final}
Let~$\prec$ be an elimination order for~$G$.
Then, there is at~most one clause gadget $\Clause(c)$ such that $\final \prec V(\Clause(c))$.
\end{lemma}
\begin{proof}
The vertex~$\final$ is adjacent to all the vertices~$c_\bot$.
By construction, no vertex of $G$ apart from~$\final$ has more than one $c_\bot$ in its closed neighborhood.
Let $w$ be a~1-twin of $\final$ in $G_{\final}$.
Assume for the sake of contradiction that $c_1$ and $c_2$ are two clauses of $\varphi$ with $\final \prec V(\Clause(c_1)) \cup V(\Clause(c_2))$.
As $\final$ is adjacent to both $(c_1)_\bot$ and $(c_2)_\bot$, $w$ must be either a~vertex of $\Clause(c_1)$ or a~vertex of $\Clause(c_2)$.
In both cases, it still has four~neighbors in $G_{\final}$ inside its clause gadget, thus $w$ cannot be a 1-twin of $\final$; a~contradiction.
\end{proof}

Next we formalize the dependence between eliminating the variable gadgets and eliminating the clause gadgets.

\begin{lemma}\label{lem:clause inp}
Let~$\prec$ be an elimination order for~$G$.
Let~$c = l_1 \lor l_2 \lor l_3$ be a~clause of~$\varphi$.
Let~$v$ be the first eliminated vertex of~$\Clause(c)$.
Then either~$\final \prec v$ or there is an~$l \in \{l_1, l_2, l_3\}$ such that for any vertex~$y$ in~$\Var(l)$ such that~$N(y)$ contains the (open) neighborhood of the representative of~$l$, we have~$y \prec v$.
\end{lemma}
\begin{proof}
Let $w$ be a~1-twin of~$v$ in~$G_v$.
Note that~$G_v$ still contains the whole gadget~$\Clause(c)$, and so~$N_{G_v}[v]$ contains~$V(\Clause(c))$.
As all the vertices of~$G$ apart from those of~$\Clause(c)$ have at most two neighbors in~$\Clause(c)$, $w$ has to be a~vertex of~$\Clause(c)$.
Assume, for the sake of contradiction, that~$v \prec \final$ and for each~$l \in  \{l_1, l_2, l_3\}$, there is a~vertex~$y_l$ of~$\Var(l)$ with~$v \prec y_l$ that is adjacent to the whole open neighborhood of the representative of~$l$. We consider three cases.

{\bf Case 1:}
$v$ is~$c_\top$. Then, by construction of~$G$,~$v$ is adjacent to~$y_{l_1}, y_{l_2}$, and~$y_{l_3}$, and so~$w$ must be adjacent to at least two vertices among~$y_{l_1}, y_{l_2}$ and~$y_{l_3}$.
But no vertex in~$w\in V(\Clause(c)) \setminus \{c_\top\}$ has this property.

{\bf Case 2:}
$v$ is of the form~$c_{l}$ for~$l \in \{l_1, l_2, l_3\}$. Without loss of generality, assume $v = c_{l_1}$.
Then vertex~$v$ is adjacent to~$y_{l_1}$ and not adjacent to~$y_{l_2}$ and~$y_{l_3}$.
Thus~$w$ cannot be~$c_\top$, nor~$c_{l_2}$, nor~$c_{l_3}$.
Vertex $w$ cannot be $c_\bot$ either, since it is still adjacent to $\final$.

{\bf Case 3:}
$v=c_\bot$.
By the previous cases $c_\bot$ has no 1-twins in $G_v$.
\end{proof}

\begin{lemma}\label{lem:var last}
Let $\prec$ be an elimination order for~$G$.
Let~$l$ be a literal appearing in clauses~$c_1, \ldots, c_p$, and let~$v_l$ be its representative.
Assume that the last eliminated vertex from~$\Var(l)$, say~$v$, is fully adjacent to~$N_G(v_l)$.
Then, for each $i \in [p]$, there is a~vertex~$x_i$ of the gadget~$\Clause(c_i)$ with~$x_i \prec v$. 
\end{lemma}
\begin{proof}
Assume, for the sake of contradiction, that there is a~clause~$c$ among~$c_1, \ldots, c_p$ with all its vertices in~$G_v$.
Let~$w$ be a~1-twin of~$v$ in~$G_v$.
Then~$v$ has exactly two neighbors in~$\Clause(c)$:~$c_\top$ and~$c_l$.
This implies that~$w$ cannot be a~vertex of~$\Clause(c)$, as $w$ still has four neighbors in~$V(\Clause(c))$.
Vertex~$w$ cannot be in another clause gadget, since it would then have no neighbors in~$V(\Clause(c))$.
Vertex~$w$ is not~$\final$ nor~$\isolated$ since they are not adjacent to~$c_\top$ and~$c_{l}$.

Therefore,~$w$ is a vertex of another variable gadget, and it is adjacent to~$c_\top$.
Then~$w$ is adjacent to a~vertex~$c_{l'}$ of~$\Clause(c)$, with $l' \neq l$.
Thus~$v$ and~$w$ are not 1-twins; a~contradiction.
\end{proof}

We can now turn an elimination order into a variable assignment that satisfies all the clauses but at~most one, which together with \cref{lem:sat-to-order}, proves~\cref{thm:sdd-npc}.

\begin{lemma}\label{lem:order-to-sat}
If~$G$ admits an elimination order, then all clauses of~$\varphi$ but at~most one can be satisfied.
\end{lemma}
\begin{proof}
Let~$\prec$ be an elimination order for~$G$.
For each variable~$x$, let~$v_x$ be the last vertex eliminated from the variable gadget~$\Var(x)$.
If~$N_G(v_x)$ contains~$N_G(x_0)$ then we set~$x$ to true, otherwise we set~$x$ to false. This gives a~truth assignment~$\mathcal A(\prec)$.

Assume for the sake of contradiction that there are two clauses~$c_1, c_2$ not satisfied by~$\mathcal A(\prec)$.
Assume that~$v \in \Clause(c_1)$ is the first eliminated vertex among~$\Clause(c_1) \cup \Clause(c_2)$.
By \cref{lem:final}, the vertex~$\final$ is in~$G_v$.
Thus by \cref{lem:clause inp}, there is a literal~$l$ appearing in~$c_1$ such that any vertex in~$\Var(l)$ adjacent to the neighborhood of the representative of~$l$ is eliminated before~$v$. 
By definition of~$\mathcal{A}(\prec)$, this implies that $\Var(l)$ is eliminated before any vertex of~$\Clause(c_1)$ has been eliminated, which is a contradiction by \cref{lem:var last}.
\end{proof}

\newcommand{\etalchar}[1]{$^{#1}$}

\end{document}